\documentclass[a4paper, runningheads, USenglish,cleveref, autoref, thm-restate]{llncs}

\usepackage{microtype}
\usepackage{graphicx}
\usepackage{booktabs} 
\usepackage[T1]{fontenc}

\usepackage{hyperref}

\usepackage{amsmath}
\usepackage{mathtools}
\usepackage{thmtools}
\usepackage{amssymb}

\usepackage[textsize=tiny]{todonotes}

\usepackage{xspace}
\usepackage{framed}
\usepackage{xifthen}
\usepackage{tikz}
\usetikzlibrary{calc}

\usepackage{algorithmic}
\usepackage{algorithm}

\usepackage{enumitem}
\usepackage{thm-restate}

\newcommand{\yes}{\textsc{Yes}\xspace}
\newcommand{\no}{\textsc{No}\xspace}
\newcommand{\yesinstance}{\yes-instance\xspace}
\newcommand{\noinstance}{\no-instance\xspace}

\newcommand{\positives}{\ensuremath{\mathbb{N}^{+}}\xspace}
\newcommand{\ttuple}{\ensuremath{\mathbb{T}}\xspace}

\newcommand{\NPH}{\textsf{NP}-hard\xspace}

\newcommand{\vc}{\textsc{Vertex Cover}\xspace}
\newcommand{\VC}{\vc}
\newcommand{\fvs}{\textsc{Feedback Vertex Set}\xspace}
\newcommand{\FVS}{\fvs}
\newcommand{\fairvc}{\textsc{$\mathbb{T}$-Fair Vertex Cover}\xspace}

\newcommand{\coloursum}{\ensuremath{(\sum_{i=1}^{t}k_{i})}\xspace}

\newlength{\RoundedBoxWidth}
\newsavebox{\GrayRoundedBox}
\newenvironment{GrayBox}[1]%
   {\setlength{\RoundedBoxWidth}{.93\textwidth}
    \def\boxheading{#1}
    \begin{lrbox}{\GrayRoundedBox}
       \begin{minipage}{\RoundedBoxWidth}}%
   {   \end{minipage}
    \end{lrbox}
    \begin{center}
    \begin{tikzpicture}%
       \node(Text)[draw=black!20,fill=white,rounded corners,%
             inner sep=2ex,text width=\RoundedBoxWidth]%
             {\usebox{\GrayRoundedBox}};
        \coordinate(x) at (current bounding box.north west);
        \node [draw=white,rectangle,inner sep=3pt,anchor=north west,fill=white] 
        at ($(x)+(6pt,.75em)$) {\boxheading};
    \end{tikzpicture}
    \end{center}}     
\newenvironment{defproblemx}[2][]{\noindent\ignorespaces%
                                \FrameSep=6pt%
                                \parindent=0pt%
                \vspace*{-1.5em}
                \ifthenelse{\isempty{#1}}{%
                  \begin{GrayBox}{\textsc{#2}}%
                }{%
                  \begin{GrayBox}{\textsc{#2}  parameterized by~{#1}}%
                }
                \begin{tabular*}{\textwidth}{@{\hspace{.1em}} >{\itshape} p{1.8cm} p{0.8\textwidth} @{}}%
            }{
                \end{tabular*}%
                \end{GrayBox}%
                \ignorespacesafterend
            }  

\newcommand{\defproblemshort}[3]{
  \vspace{1mm}
\begin{center}
\noindent\fbox{
  \begin{minipage}{.9\linewidth}
  \begin{tabular*}
    {\linewidth}{@{\extracolsep{\fill}}lr} \textsc{#1}   \\ \end{tabular*}
  {\bf{Input:}} #2  \\
  {\bf{Task:}} #3
  \end{minipage}
  }
\end{center}
  \vspace{1mm}
}

\newcommand{\NP}{\textsf{NP}}

\newcommand{\twidth}{\ensuremath{t_w}}

\newcommand{\card}[1]{\ensuremath{{\vert {#1} \vert }}}
\newcommand{\set}[1]{\ensuremath{\left\{ {#1} \right\}}}
\newcommand{\fn}[3]{\ensuremath{{{#1} : {#2} \rightarrow {#3}}}}

\newcommand{\cO}{\mathcal{O}}

\newcommand{\ffvsfull}{\textsc{$\mathbb{T}$-Fair Feedback Vertex Set}}
\newcommand{\ffvs}{\textsc{$\mathbb{T}$-Fair FVS}}

\newcommand{\ins}{\ensuremath{\mathtt{ins}}}
\newcommand{\proj}{\ensuremath{\mathtt{proj}}}
\newcommand{\join}{\ensuremath{\mathtt{join}}}
\newcommand{\glue}{\ensuremath{\mathtt{glue}}}

\newcommand{\fvsruntime}{\ensuremath{2^{O(\twidth)} n^{\cO(1)}}}
\newcommand{\fvsalgo}{{\sf Algo-Fair-FVS-Treewidth}}
\newcommand{\fvsalgofpt}{{\sf Algo-Fair-FVS-TCB}}

\title{Addressing Bias in Algorithmic Solutions:
Exploring Vertex Cover and Feedback Vertex Set}

\author{Sheikh Shakil Akhtar\inst{1}\orcidID{ 0000-0002-3653-6670 } \and
	Jayakrishnan Madathil\inst{2}\orcidID{0000-0001-6337-6759} \and
	Pranabendu Misra\inst{1}\orcidID{0000-0002-7086-5590} \and
	Geevarghese Philip\inst{1}\orcidID{0000-0003-0717-7303}}

\authorrunning{S.S. Akhtar,J. Madathil, P. Misra, G. Philip}

\institute{Chennai Mathematical Institute \and University of Glasgow}

\begin{document}

\maketitle

\begin{abstract}

A typical goal of research in combinatorial optimization is to come up with fast
algorithms that find optimal solutions to a computational
problem.
The process that takes a real-world problem and extracts a clean mathematical
abstraction of it often throws out a lot of ``side information'' which is deemed
irrelevant. However, the discarded information could be of real significance to
the end-user of the algorithm's output. All solutions of the same cost are not
necessarily of equal impact in the real-world; some solutions may be much more
desirable than others, even at the expense of additional increase in cost. If the impact, positive or negative, is mostly felt by some specific (minority) subgroups of the population, the population at large will be largely unaware of it.

In this work we ask the question of finding solutions to combinatorial optimization problems that are ``unbiased'' with respect to a collection of specified subgroups of the total
population. We consider a simple model of bias, and study it via two basic optimization problems on graphs: 
\VC and \FVS, which are both NP-hard.
Here, the input is a graph and the solution is a subset of the vertex set. The vertices represent members of a population, and
each vertex has been assigned a subset of colors, where each color indicates membership of a specific subgroup. The goal is to find a small-sized solution to the optimization problem in which no color appears more than a specified---per-color---number of times.
%
%
The colors can be used to model various
relevant---economic, political, demographic, or other---classes to which the
entities belong, and the variants that we study can then be used to look for
small solutions which enforce per-class upper bounds on the number of removed
entities. These upper-bounds enforce the constraint that no subclass of the population is over-represented in the solution.

We show the new variants of \VC and \FVS, obtained by adding these additional constraints, are \emph{Fixed-Parameter Tractable}, when parameterized by various combinations of the solution size, the number of colors, and the treewidth of the graph. Our results shows that it is possible to devise fast algorithms to solve these problem in many practical settings.

\end{abstract}


\section{Introduction}\label{mainintroduction}

Consider a hospital that treats a large number of patients every day. The patients require timely access to a number of diagnostic and medical procedures for making a successful recovery. Since the manpower and equipment at a hospital are shared and in high-demand, scheduling these procedures is a
non-trivial task. A simple way to schedule the patients could be as follows:
construct a graph $G$ where the vertices $V(G)$ represent the patients, and we
have an edge between patient $u$ and patient $v$ if they have a common medical
procedure; i.e. both can't be scheduled simultaneously. The objective is to
maximize the number of patients that can be scheduled now, or equivalently to
minimize the number of patients who will have to wait for later.
The second formulation of the problem can be recognized as the classic {\sc Vertex Cover} problem, for
which efficient Parameterized algorithms~\cite{FPTbook2015} and Approximation
algorithms~\cite{ApproxBookWS2011,ApproxBookVazirani2001} are well-known.

The known algorithms for {\sc Vertex Cover} however aren't designed for being
\emph{unbiased} and \emph{impartial}; i.e. it can't be ensured that the patients
who are rescheduled are not disproportionately from one subgroup, e.g. economically weaker. Here the
subgroups may be social, economic, medical or as per some other relevant criteria. If not addressed, such (unintended) bias can have severe and lasting impact.
%
Moreover, if the adverse effects are limited
to some small subgroups, then the majority population will be quite unaware that there is a problem at all. 
For example, if the local hospital is resource constrained and the delayed patients are largely from some minority social
groups, then to the general population it will appear that the hospital is well-functioning and there is no need for any additional funding or resources! This is clearly a serious problem.

It is self-evident that (combinatorial) algorithms themselves are unbiased. The
bias in the solutions computed by them appears in other ways. For example, the order in which the input data is presented could influence the output. Sometimes the bias may be inherent in the data itself; e.g. the prevalence and severity of
certain medical conditions varies by gender, ethnicity, income, age etc. This means that certain sophisticated medical diagnostics and procedures may correlate with certain subgroups of patients. The algorithm may decide to schedule these patients for later in the pursuit of an optimal solution.

One of the main reasons for this un-indented bias is that when we reduce a real-world problem to an
abstract mathematical problem, we try to simplify it as much as possible and discard a lot of associated information. 
Furthermore, certain heuristics employed to speed up the algorithm could have unexpected consequences. For example, most algorithms for {\sc Vertex Cover} will pick high-degree vertices into the solution. These vertices corresponds to patients who require a number of medical procedures and are likely to be in a critical condition, and must not be delayed!

Other examples arise in the applications of resource allocation using
algorithms.
Consider the process of setting up sample-collection centers across a city
medical-tests during the COVID-19 pandemic~\cite{munguia2021fair}. Due to cost reasons, only a limited
number of such centers can be opened and the goal will be to place them so as to
minimize the maximum travel distance of a citizen to the nearest center. 
Similar challenges arise in many other public health settings.
The problem of finding the optimal location for these centers can be modeled as
a \emph{Clustering} problem which is very well-studied (see, e.g.
~\cite{charu2013data,tan2016introduction}). However, consider the subset of 
people who are all quite far way from the centers, as compared to the rest of the population, e.g. in a remote settlement.
They are less likely to travel to the centers compared to the general population. If they happen to be disproportionately from a certain subgroup, which could be case for a historically disadvantaged ethnic subgroup, then the data collected from the tests will fail to properly account for this subgroup. Then, any policy built upon this data will be flawed~\cite{abbasi2020risk}.

The current algorithms for these computational problems are oblivious to the biases in the output solution.
If unchecked, such bias can have significant impact on certain groups. 
We emphasize once more that this is not because the algorithms are themselves
biased, but because of the attributes of the problem instance itself, the choice
of mathematical abstraction, the presentation of the input data etc. 
We also remark that simple strategies such as randomizing the order of input data may or may not be helpful; 
it is unclear without a formal (or experimental) analysis. Moreover, this will not ameliorate
every form of bias (such as the consequence of the high-degree rule for {\sc
  Vertex Cover} mentioned earlier). It is therefore essential to consider more concrete ways of
addressing the bias in algorithmic solutions.

Motivated by this, we explore the class of
computational problems derived from classic NP-complete optimization problems obtained by introducing
additional constraints for being unbiased with respect to population subgroups.
This new class of problems are natural combinatorial questions
that are interesting in their own right. For simplicity, we focus on graph
problems where the vertices represent members of the general population, and
vertex subsets naturally represent various subgroups of the population. 
We may also generalize this definition to broader classes of problems such as those modeled by \emph{Constraint Satisfaction Problems(CSPs)} or \emph{Linear Programming}, but at the expense of clarity. Therefore we state them for optimization problems on graphs, which capture a broad class of classic optimization problems.
 Let $\Pi$ be a
graph problem such that the solution to \(\Pi\) is a subset of vertices of the
input graph $G$. Then, we have the following \emph{unbiased} variant of $\Pi$: 

\defproblemshort{UnBiased $\Pi$}{$(G,c)$ where $G$ is a graph whose vertices are
  labeled by \emph{colors} from a set $\{1,2,\ldots, t\}$ via the function
  $c:V(G) \rightarrow 2^{\{1,2,\ldots, t\}} \setminus \{\emptyset\}$.}{Find a
  solution $S$ of minimum size such that for every $i \in \{1,2,\ldots,t\}$,
  $\frac{|\{v \in S ~|~ i \in c(v)\}|}{|S|} = \frac{|\{v \in V(G) ~|~ i \in
    c(v)\}|}{|V(G)|}$}

Here the colors $\{1,2,\ldots,t\}$ can represent any relevant subgroup of the
total population. A single vertex can be a member of more than one subgroup,
which is captured by the \emph{coloring function $c$} ; and note that every vertex is
part of at least one subgroup. Our objective is to compute a solution $S$ where
the fraction of vertices for each color $i$ is the same as the fraction of the
color in the total population. We call such a solution \emph{unbiased}.
Naturally, our objective is to compute an unbiased solution of the smallest
size.

We remark that the \emph{price of fairness} in a given instance $(G,c)$ of {\sc UnBiased $\Pi$} could be very high, due to the strict unbiasedness constraints.
Indeed, it is possible to construct artificial examples where the difference between the optimal solution to $G$ and the optimal \emph{unbiased} solution to $(G,c)$ is unbounded. 
Therefore, we define the following more general variant. Here $\alpha \leq 1 \leq \beta$ are real numbers, that set the desired level of \emph{unbiased-ness}.

\defproblemshort{$(\alpha, \beta)$-UnBiased $\Pi$}{$(G,c)$ where $G$ is a graph whose vertices are
    labeled by \emph{colors} from a set $\{1,2,\ldots, t\}$ via the function
    $c:V(G) \rightarrow 2^{\{1,2,\ldots, t\}} \setminus \{\emptyset\}$.}{Find a
    solution $S$ of minimum size such that for every $i \in \{1,2,\ldots,t\}$, we have 
    $\alpha \cdot \frac{|\{v \in V(G) ~|~ i \in c(v)\}|}{|V(G)|} \leq \frac{|\{v \in S ~|~ i \in c(v)\}|}{|S|} \leq \beta \cdot \frac{|\{v \in V(G) ~|~ i \in c(v)\}|}{|V(G)|}$}

Observe that for $\alpha = 0, \beta = \infty$ we obtain $\Pi$, whereas for $\alpha = \beta = 1$, we obtain {\sc UnBiased $\Pi$}. 
We can obtain the desired level of trade-off between the level of bias and the price of fairness by setting $\alpha$ and $\beta$.

In this paper we consider these problems in the framework of Parameterized
Complexity~\cite{FPTbook2015}. Here, the problem instances are parameterized,
i.e. they consist of a pair $(I,k)$, where $I$ is an instance of a problem $\Pi$
and $k$ is a number called the parameter representing some structural property of $I$ that is typically bounded in practice. A typical parameter is the optimum solution size, and obtaining an FPT algorithm for this parameter means that we can compute an optimal solution efficiently when the solution size is bounded, which is the case in many practical settings, even if the input instance itself is very large.
Other examples parameters are treewidth of a
graph, the genus of the graph etc. The objective is to obtain \emph{FPT algorithms} that find an optimal
solution in time $f(k) \cdot n^{O(1)}$, where $f$ is a function of $k$ alone.
FPT algorithms are employed to compute optimal solutions to NP-hard problems in
nearly-polynomial time. Another objective is to design \emph{kernelization
  algorithms}, that given an instance $(I,k)$, run in polynomial time, and
produce an an equivalent instance $(I',k')$ such that $|I'| + k' = k^{O(1)}$.
$(I',k')$ is called a \emph{polynomial kernel}. Kernelization captures the
notion of efficient data-reduction algorithms. We refer to~\cite{FPTbook2015}
for the details.

We define a slightly different variant of $(\alpha, \beta)$-UnBiased $\Pi$ which naturally fits
into this paradigm. Here for each color-class $i \in \{1,\ldots, t\}$, we are
given a number $k_i$, and let 
$\mathbb{T} = (k_i ~|~ 1 \leq i \leq t)$ denote
this tuple of integers. 
Then we define the following problem.

\defproblemshort{$(\alpha, \beta) \mbox{-} \mathbb T$-fair $\Pi$}{$(G,c)$ where $G$ is a graph whose
  vertices are labeled by \emph{colors} from a set $\{1,2,\ldots, t\}$ via the
  function $c:V(G) \rightarrow 2^{\{1,2,\ldots, t\}} \setminus \{
  \emptyset \}$, and a \(t\)-tuple of
  integers \({\mathbb T} = (k_{1}, k_{2},\dots,k_{t})\).}
  {Decide whether there is a
  solution $S$ such that $|S| \leq k$ and $\alpha \cdot k_i \leq |\{v \in S ~|~ i \in c(v)\}| \leq \beta \cdot k_i$ for each $1 \leq i
  \leq t$. Here $k = \sum_{i=1}^t k_i$.} 

Observe that, if $k$ were the size of the solution $S$, then we can set $k_i = k \cdot \frac{|\{v \in V(G) ~|~ i \in c(v)\}|}{|V(G)|}$
to obtain  {\sc $(\alpha, \beta)$-UnBiased $\Pi$}. Also, when $\alpha = \beta = 1$, we call it {\sc $\mathbb T$-fair $\Pi$}, which corresponds to 
{\sc Unbiased $\Pi$}.

For simplicity, we say $(G,c)$ is \emph{$t$-colored} if vertices of $G$ are
labeled by \emph{colors} from the set $\{1,2,\ldots, t\}$ via the given function
$c:V(G) \rightarrow 2^{\{1,2,\ldots, t\}} \setminus \{\emptyset\}$. We say that
a solution $S$ is \emph{$(\alpha, \beta) \mbox{-} \mathbb T$-fair} if it is a solution to $\Pi$, $|S| \leq k$ and $\alpha \cdot k_i \leq |\{v \in S ~|~ i \in c(v)\}| \leq \beta \cdot k_i$ for each $1 \leq i \leq t$. In the strict setting where $\alpha = \beta = 1$, we call it \emph{$\mathbb T$-fair solution}.

In this paper, we initiate the study of these class of problems via two classic graph problems, namely {\sc Vertex Cover}
and {\sc Feedback Vertex Set}, which are defined as follows. Recall that a
vertex-subset $S$ is a \emph{vertex cover} of a graph $G$ if $G-S$ has no edges.
$S$ is a \emph{$(\alpha, \beta) \mbox{-} \mathbb{T}$-fair vertex cover} if it satisfies the constraints
imposed by $(G,c)$ and $\mathbb{T}$. A vertex-subset $S$ is a \emph{feedback
  vertex set} of a graph $G$, if $G-S$ is acyclic. We can similarly define a
\emph{$(\alpha, \beta) \mbox{-} \mathbb{T}$-fair feedback vertex set}.

\defproblemshort{$(\alpha, \beta) \mbox{-} \mathbb T$-Fair Vertex Cover ($\mathbb{T}$-Fair VC)}{A
  \(t\)-coloured graph \((G,c)\) and a \(t\)-tuple of integers \({\mathbb T} =
  (k_{1}, k_{2},\dots,k_{t})\).}{Decide whether \(G\) has a $(\alpha, \beta) \mbox{-} \mathbb T$-fair
  vertex cover.}

\defproblemshort{$(\alpha, \beta) \mbox{-} \mathbb T$-Fair Feedback Vertex Set($\mathbb{T}$-Fair FVS)}{A
  \(t\)-coloured graph \((G,c)\) and a \(t\)-tuple of integers \({\mathbb T} =
  (k_{1}, k_{2},\dots,k_{t})\).}{Decide whether \(G\) has a $(\alpha, \beta) \mbox{-} \mathbb T$-fair
  feedback vertex set.}

We first consider {\sc $\mathbb{T}$-Fair VC} and {\sc $\mathbb{T}$-Fair FVS} and 
obtain FPT-algorithms and polynomial kernels for these problems when
parameterized by solution-size, number of colors, and the treewidth of the
underlying graph. Note that the treewidth of a graph is never larger than the
size of a minimum vertex cover, and at most one larger than the size of a
minimum feedback vertex set~\cite{FPTbook2015}. 
We then show that these algorithms can be extended to the more general {\sc $(\alpha, \beta) \mbox{-} \mathbb{T}$-Fair VC} 
and {\sc $(\alpha, \beta) \mbox{-} \mathbb{T}$-Fair FVS}.
Formally, 

\begin{restatable}{theorem}{fairVCfpt}
  \label{thm:fairVCfpt}
  \fairvc can be solved in \(\cO(n^{2}t(\prod_{i=1}^{t}k_{i}^{2})\cdot
  1.4656^{k})\) time where \(n\) is the number of vertices in the input graph
  \(G\), \(t\) is the number of colours used by the colouring function \(c\),
  \((k_{1}, k_{2}, \dotsc, k_{t})\) is the colour budget specified in the input,
  and \(k = \coloursum\).
\end{restatable}

\begin{restatable}{theorem}{fairVCkernel}
  \label{thm:fairVCkernel}
	There is a polynomial time algorithm which outputs a polynomial kernel for
  \fairvc. If \(G\) is the input graph, \(t\) is the number of colours used by
  the colouring function \(c\), and \((k_{1}, k_{2}, \dotsc, k_{t})\) is the
  colour budget specified in the input, then the size of this kernel is
  quadratic in \(k = \coloursum\).
\end{restatable}

\sloppy

\begin{restatable}{theorem}{fairVCtreewidth}
  \label{thm:fairVCtreewidth}
	Given a graph \(G\), a nice tree decomposition \((\mathcal{T}, (B_x)_{x \in
    V(\mathcal{T})})\) of \(G\) of width \twidth\(~\)and with \(l\) nodes, a
  colouring function \(c\), where \(t\) is the number of colours used by the
  colouring function \(c\), and \((k_{1}, k_{2}, \dotsc, k_{t})\) is the colour
  budget specified in the input, there is an algorithm which solves \fairvc in
  time \(O(tl {\twidth}^2 2^{\twidth + 1} \prod_{i = 1} ^t k_i ^2)\).
\end{restatable}

\fussy

\begin{restatable}{theorem}{fairFVStreewidth}
  \label{thm:fairFVStreewidth}
	Given a graph \(G\) with \(n\) vertices, a nice tree decomposition \break
	\((\mathcal{T},(B_x)_{x \in V(\mathcal{T})})\) 
  of \(G\) of width \twidth\(~\)and with \(l\)
  nodes, a colouring function \(c\), where \(t\) is the number of colours used
  by the colouring function \(c\), and \((k_{1}, k_{2}, \dotsc, k_{t})\) is the
  colour budget specified in the input, there is an algorithm which solves
  \ffvsfull\(~\)in time \(n^{O(1)} 2^{O(\twidth)}\).
\end{restatable}

\begin{restatable}{theorem}{fairFVSfpt}
  \label{thm:fairFVSfpt}
	\ffvsfull\(~\)can be solved in \(n^{O(1)} 2^{O(k)}\)
  time where \(n\) is the number of vertices in the input graph
  \(G\), \(t\) is the number of colours used by the colouring function \(c\),
  \((k_{1}, k_{2}, \dotsc, k_{t})\) is the colour budget specified in the input,
  and \(k = \coloursum\).
\end{restatable}

From the above we obtain the following results.

\begin{restatable}{theorem}{vcAlphaBetaThm}\label{thm:vc-alpha-beta}
	For every fixed $\alpha, \beta$, the $(\alpha, \beta) \mbox{-} \mathbb{T}$-{\sc Fair VC} problem can be solved in time $\cO(n^{2}t(\prod_{i=1}^{t}k_{i}^{2}((\beta - \alpha) k_i + 1))\cdot 1.4656^{k})$.
\end{restatable}

\begin{restatable}{theorem}{fvsAlphaBetaThm}\label{thm:vc-alpha-beta}
	For every fixed $\alpha, \beta$, the $(\alpha, \beta) \mbox{-} \mathbb{T}$-{\sc Fair FVS} problem can be solved in time $\prod_{i=1}^{t} ((\beta - \alpha) k_i + 1) \cdot n^{O(1)} 2^{O(k)}$.
\end{restatable}

More generally, for any classic computation problem $\Pi$ such that {\sc $\mathbb T$-Fair $\Pi$} has an algorithm with runtime $g(k_1, k_2,\ldots, k_t, n)$, we have the following.
\begin{restatable}{theorem}{alphaBetaThm}\label{thm:alpha-beta}
	For every fixed $\alpha, \beta$, the $(\alpha, \beta) \mbox{-} \mathbb{T}$-{\sc Fair $\Pi$} problem can be solved in time $\Pi_{i \in [t]} ((\beta - \alpha) k_i + 1)) \cdot g(\beta k_1, \beta k_2,\ldots, \beta k_t, n)$, where $g(k_1, k_2,\ldots, k_t, n)$ is the runtime of an algorithm for $\mathbb{T}$-{\sc Fair $\Pi$}. 
\end{restatable}

\section{Related Work}\label{relatedwork}

We remark that in this paper we study a very simple model of algorithmic bias that easily extends to very broad classes of optimization problems and algorithms, e.g. problems that can be modeled by LPs or CSPs. This leads to a new class of optimization problems, and our objective here is to initiate the study of this class of problems, by showing that it is possible to obtain non-trivial algorithms for them.

Some of these works gave FPT algorithms combinatorial problems under certain notions of fairness:
Parameterized complexity of fair feedback vertex set problem
\cite{KANESH20211},
parameterized complexity of fair deletion problems
\cite{10.1007/978-3-319-55911-7_45},
parameterized complexity of fair vertex evaluation problems~\cite{knop2019parameterized}.
We do not attempt to summarize such a large body of work here. However, we must acknowledge that  we were motivated by these works and have attempted to define a very simple notion of algorithmic fairness that can be applied to a broad class of problems, and how we can address the unintended bias hidden in widely-used and standard algorithmic techniques (e.g. the high-degree rule for \vc).

\section{Preliminaries}\label{sec:preliminaries}

We use \(2^{X}\) to denote the power set of a set \(X\) and \(\lvert X \rvert\) to denote its cardinality. For \(r \in
\positives\) we use \([r]\) to denote the set \(\{1,2,\dotsc,r\}\). All our
graphs are finite, simple and undirected, unless specified otherwise. We use
\(V(G)\) and \(E(G)\), respectively, to denote the vertex and edge sets of a
graph \(G\). For a graph \(G\) and vertex \(v\in V(G)\) we use (i) \(N(v)\) to
denote the \emph{open neighbourhood} \(\{w \in (V(G) \setminus
\{v\})\;\mid\;(v,w)\in E(G)\}\) of \(v\), and (ii) \(N[v]\) to denote the
\emph{closed neighbourhood} \((\{v\} \cup N[v])\) of \(v\). The \emph{degree} of a vertex \(v\) of \(G\), denoted by \(deg(v)\) is the number \(\lvert N(v) \rvert\), i.e., the number of vertices in \(G\) which are adjacent to \(v\).

We use standard notations from graph theory. For an integer \(t\in\positives\) and a graph \(G\), a \emph{\(t\)-colouring
  function of \(G\)} is any function \(c: V \to (2^{[t]} \setminus
\{\emptyset\})\) that assigns at least one ``colour'' from the set \([t]\) to
each vertex of \(G\). A \emph{\(t\)-coloured graph} is a pair \((G,c)\) where
\(c\) is a \(t\)-colouring function of \(G\). For a \(t\)-coloured graph
\((G,c)\), vertex \(v\in V(G)\), and \(i\in[t]\), we use \(N_{i}(v)\) to denote
the \emph{colour-\(i\) neighbourhood} \(\{w \in N(v)\;\mid\;i \in c(w) \}\) of
the set of all neighbours of \(v\) which have been assigned the colour \(i\) by
\(c\). For a subset \(S \subseteq V(G)\) of vertices we use \(c_{i}(S)\) to
denote the \emph{number} \(\lvert \{v \in S | i \in c(v) \} \rvert\) of vertices
in \(S\) which have been assigned the colour \(i\) by \(c\). For a \(t\)-tuple
of integers \(\ttuple = (k_{1}, k_{2},\dots,k_{t})\) we say that \(S \subseteq
V(G)\) is \emph{\ttuple-fair} if for each \(i \in [t]\) it is the case that
\(c_{i}(S) = k_{i}\) holds: the number of vertices in \(S\) which have the
colour \(i\), is exactly \(k_{i}\).
For a fixed finite alphbet \(\Sigma\), a \textit{paramaterised problem} is a language \(L \subseteq \Sigma^* \times \mathbb{N}\). For an instance \((x, k) \in \Sigma^* \times \mathbb{N}\), \(k\) is called the \textit{parameter}. A parameterised problem \(L \subset of \Sigma^* \times \mathbb{N}\) is called \textit{fixed parameter tractable} if there exists an algorithm, say \(\mathcal{A}\), a computable function \(f : \mathbb{N} \rightarrow \mathbb{N}\), and a constant \(c\) such that, given \((x, k) \in \Sigma^* \times \mathbb{N}\), the algorithm correctly decides whether \((x, k) \in L\) in time bounded by \( \lvert (x , k) \rvert ^c f(k)\). The complexity class of all fixed-parameter tractable problems is called FPT.

We now move on to the notion of \textit{tree decomposition} of a graph. A tree decomposition of a graph \(G\) is a pair \((\mathcal{T}, (B_x)_{x \in V(\mathcal{T})})\), where \(\mathcal{T}\) is a tree whose every node \(x\) is assigned a vertex subset \(B_x \subseteq V(G)\), called a bag, such that the following conditions hold:

\begin{itemize}
	\item \(\bigcup_{x \in V(\mathcal{T})} B_x = V(G)\). In other words, every vertex of \(G\) is in at least one bag.
	\item For every edge \(uv \in E(G)\), there exists a node \(x\) of \(\mathcal{T}\) such that the bag \(B_x\) contains both \(u\) and \(v\).
	\item For every \(u \in V(G)\), the set \(\mathcal{T}_u = \{x \in V(\mathcal{T}) : u \in B_x\}\), i.e., the set of nodes whose corresponding bags contain \(u\), induces a connected subtree of \(\mathcal{T}\).
\end{itemize}

 The \textit{width} of a tree decomposition \((\mathcal{T}, (B_x)_{x \in V(\mathcal{T})})\) equals 

\(\text{max}_{x \in V(\mathcal{T})} \lvert B_x \rvert - 1\), that is the maximum size of its bag minus \(1\). The \textit{treewidth} of a graph \(G\) is the minimum possible width of a tree decomposition \(G\). 
For the purposes of designing algorithms, we will use a special kind of tree decomposition of a graph, known as \textit{nice tree decomposition}. A tree decomposition \((\mathcal{T}, (B_x)_{x \in V(\mathcal{T})})\) of a graph \(G\) is said to be a nice tree decomposition of \(G\) if the following conditions are satisfied:

\begin{itemize}
	\item \(\mathcal{T}\) has a root node
	\item \(B_x = \emptyset\), when \(x\) is either the root node or a leaf node. In other words, the root and all the leaves of \(\mathcal{T}\) have empty bags.
	\item Every non-leaf node of \(\mathcal{T}\), say \(x\) is one of the following four types:
	\begin{enumerate}
		\item \textbf{Introduce Vertex Node:} \(x\) has exactly one child \(y\), such that \(B_x = B_y \cup \{v\}\), for some \(v \notin B_y\). We say that \(v\) is \textit{introduced} at \(x\).
		\item \textbf{Forget Vertex Node:} \(x\) has exactly one child \(y\), such that \(B_x = B_y \setminus \{v\}\), for some \(v \in B_y\). We say that \(v\) is \textit{forgotten} at \(x\).
		\item \textbf{Introduce Edge Node:} \(x\) has exactly one child \(y\), such that \(B_x = B_y\) and for some edge \(uv \in E(G)\) with \(\{u, v\} \subseteq B_x(= B_y)\), \(x\) is labelled with the edge \(uv\). We say that the edge \(uv\) is \textit{introduced} at \(x\). In addition, we require that every edge in \(E(G)\) in introduced \textit{exactly once}.
		\item \textbf{Introduce Vertex Node:} \(x\) has exactly two children \(y\) and \(z\), such that \(B_x = B_y = B_z\).
	\end{enumerate}
\end{itemize}

We derive algorithms based on (nice) tree decompositions \cite{FPTbook2015,bodlaender2015deterministic}. If a graph \(G\) has a tree decomposition \((\mathcal{T}, (B_x)_{x \in V(\mathcal{T})})\) of width at most \(t_w\), then in time \(O(t_w ^2 \times \text{max} (\lvert V(\mathcal{T}) \rvert, \lvert V(G) \rvert ))\) we can obtain a nice tree decomposition of \(G\) of width at most \(t_w\) that has at most \(O(t_w \lvert V(G) \rvert )\) nodes. Also if
\(G\) is a path or a cycle then one can get a nice tree decomposition
of \(G\) in time \(O(\lvert V(G) \rvert)\) with \(O(\lvert V(G) \rvert)\) nodes.
Refer to \cite{FPTbook2015} and \cite{bodlaender2015deterministic} for more information.
For the remainder of the paper, we will assume that a tree decomposition is a nice tree decomposition unless stated otherwise. For a graph \(G\) with a nice tree decomposition \\ \((\mathcal{T}, (B_x)_{x \in V(\mathcal{T})}))\), we define \(G_x = (V_x, E_x)\), the \textit{subgraph of \(G\) rooted at} \(x\) as follows:

\begin{itemize}
	\item \(V_x = B_x \cup_{y \text{ is a descendant of } x \in V(\mathcal{T})} B_y\)
	\item \(E_x = \{\text{All edges introduced at any node in the subtree of}\)
		
		\(\mathcal{T} \text{ rooted at } x \}\)
\end{itemize}

We define two \emph{Boolean} operations over the set \(\{0,1\}\).
The Boolean addition and the Boolean multiplication, denoted by
\(\oplus\) and \(\odot\) respectively are defined as follows:

\(0 \oplus 0 = 0, 0 \oplus 1 = 1 \oplus 0 = 1 \oplus 1 = 1, 
0 \odot 0 = 0 \odot 1 = 1 \odot 0 = 0, 1 \odot 1 = 1\)

\section{\fairvc}\label{sec:vc}

In this section we take up the \fairvc problem. \fairvc is \NPH by a simple
reduction from \vc, which we will now sketch. Let \((G,k)\) be an instance of
\vc. If \(G\) has fewer than \(k\) vertices then \((G,k)\) is trivially a
\yesinstance of \vc, and we construct and return a trivial \yesinstance of
\fairvc. So let us assume, without loss of generality, that \(G\) has at least
\(k\) vertices. This implies that \(G\) has a vertex cover of size at most \(k\)
if and only \(G\) has a vertex cover of size exactly \(k\).

We construct a \(1\)-coloured graph \((H,c)\) by setting (i) \(H=G\) and (ii)
\(c\) to be the function that assigns the colour set \(\{1\}\) to every vertex
of \(H\). Further, we set \(\ttuple=(k)\) to be a \(1\)-tuple with its sole
element being \(k\). \(((H,c), \ttuple)\) is the reduced instance of \fairvc.
It is easy to see that \(S\) is a vertex cover of \(G\) of size exactly \(k\) if
and only if \(S\) is a \ttuple-fair vertex cover of \(H\).

We study different parameterizations of \fairvc in this section. In
\autoref{sec:fpt} we set the parameter to be the ``total colour budget''
\coloursum. This is a natural generalization of the ``standard'' parameter \(k\)
of \vc, as suggested by the above reduction. We show---\autoref{sec:fpt}---that
\fairvc is fixed-parameter tractable with \coloursum as the parameter. We take
up the kernelization question in \autoref{sec:kernel}, and show that when the
number of colours \(t\) is a constant then \fairvc has a kernel of size
quadratic in \coloursum. In \autoref{sec:vc-treewidth} we show that \fairvc has a
single-exponential FPT algorithm when parameterized by the \emph{treewidth} of
the input graph \(G\).

\subsubsection{An FPT Algorithm Parameterized by \(k=\coloursum\)}\label{sec:fpt}

The ``total colour budget'' \(k=\coloursum\) is a natural parameter for \fairvc,
since it is an analogue of the vertex cover size of the ``plain'' \vc problem.
We show that this analogy carries over, to a good extent, to the parameterized
tractability of \fairvc: we prove \fairVCfpt

Our algorithm is an adaptation of a standard FPT algorithm for \vc that branches
on vertices of degree at least \(3\)~\cite{FPTbook2015}. The base case---graphs
with maximum degree at most \(2\)---is easy to solve in polynomial time for
``plain'' \vc, since these graphs are collections of paths and cycles. But with
the added constraint of the colour budget, it is not immediately clear that
\fairvc can be solved in polynomial time on such graphs. We get around this by
appealing to our algorithm for \fairvc on graphs of \emph{bounded treewidth};
see \autoref{sec:vc-treewidth}.

Our algorithm uses two \noinstance checks, a reduction rule, a branching rule
that applies when the graph has at least one vertex of degree \(3\), and a
subroutine for handling the base case when every vertex has degree at most
\(2\). We now describe these\footnote{See \autoref{alg:fvc} for the algorithm in
  pseudocode.}. Recall that the input consists of a graph \(G\) on \(n\)
vertices, a colouring function \(c: V(G) \to (2^{[t]} \setminus
\{\emptyset\})\), and a \(t\)-tuple of integers \(\ttuple = (k_{1},
k_{2},\dots,k_{t})\). Recall also that for any \(S\subseteq V(G)\), the
expression \(c_{i}(S)\) denotes the number of vertices in the set \(S\) which
have been assigned the colour \(i\) by the colouring function \(c\).

\paragraph{The \noinstance checks}
\begin{description}
\item[Check~1]  If \(c_{i}(V(G)) < k_{i}\) holds for at least one \(i \in [t]\)
  then return \no. 
\item[Check~2] If there is at least one pair \((i, j) \in [t] \times [t]\)
  such that (i) \(k_{i} = k_{j} = 0\) and (ii) there is an edge \((u,v) \in
  E(G)\) for which both \(i \in c(u)\) and \(j \in c(v)\) hold, then return \no.
\end{description}

\paragraph{The reduction rule}
\begin{description}
\item[Reduction Rule] If there is an \(i^{\star} \in [t] \) such that \(c_{i^{\star}}(V(G)) =
	k_{i^{\star}}\), then: Let \(X := \{v \in V(G) | i^{\star} \in c(v)\}\), be the set of all vertices in \(G\) which have the colour \(i^{\star}\) and for each \(i \in [t]\), let \(k'_i = k_i - c_i(X)\). Further, let \(c'\) be the restriction of \(c\) upon \(V(G) \setminus X\). Return the instance
  \((G - X,c',(k'_{1},\dotsc,k'_{i-1},k'_i,k'_{i+1},\dotsc,k'_{t}))\), where, $X = \{v \in V(G) | i^{\star} \in c(v)\}$, and for each $i \in [t]$, $k'_i = k_i - c_i(X)$. Otherwise return
  the---unchanged---instance \((G, c, \ttuple)\).
\end{description}

\paragraph{The branching rule}
This rule applies when \(G\) contains at least one vertex of degree at least \(3\).

\begin{description}
\item[Branching Rule] Let \(v\) be a vertex of degree at least \(3\) in
  \(G\).
Let \(H'\) be the graph obtained by deleting the vertex
    \(v\) (and its incident edges) from \(G\), let \(c'\) be the function obtained
    by restricting \(c\) to \(V(H') = (V(G) \setminus \{v\})\), and let
    \(\ttuple'=\{k_{1}',k_{2}',\dotsc,k_{t}'\}\) where for each \(i\in[t]\) we
    have: 
    \begin{itemize}
    \item \(k_{i}'=k_{i}-1\) if \(i\in c(v)\); that is, if colour \(i\) is among
      the set of colours that the \(t\)-colouring function \(c\) assigns to
      vertex \(v\), and,
    \item \(k_{i}'=k_{i}\) otherwise.
    \end{itemize}
    Let \(H''\) be the graph obtained by deleting the open neighbourhood
    \(N(v)\) (and their incident edges) from \(G\), let \(c''\) be the function
    obtained by restricting \(c\) to \(V(H'') = (V(G) \setminus N(v))\), and let
    \(\ttuple''=\{k_{1}'',k_{2}'',\dotsc,k_{t}''\}\) where for each \(i\in[t]\) we
    have \(k_{i}''=k_{i}-c_{i}(N(v))\).

    The branching rule recursively solves the two instances \((H', c',
    \ttuple')\) and \((H'', c'', \ttuple'')\). If at least one of these
    recursive calls returns \yes, then the rule returns \yes; otherwise it
    returns \no.
\end{description}

\paragraph{The subroutine for the base case}
This subroutine is applied to solve an instance of the form \((G,c,\ttuple)\)
where graph \(G\) has no vertex of degree at least \(3\). Note that such a graph
is a disjoint union of paths and cycles. Unlike in the case of \vc, we cannot
(i) directly delete isolated vertices (paths with no edges) or (ii) greedily
solve the problem on paths and cycles. This is because the solution must also
respect the colour constraints.  However, a graph which is just a disjoint union of paths and cycles has treewidth at most \(2\).

So let \(G\) be a disjoint union of paths and cycles, and let the connected
components of \(G\) be \(C_{1},C_{2},\dotsc,C_{h}\). In time linear in the size of the component we can create a nice tree decomposition for each component. Let for each, \(i \in [h]\), \(T_i\) be the nice tree decomposition we can get in linear time. Between \(T_i\) and \(T_{i + 1}\), where \(1 \leq i \leq h - 1\) we put and edge. This will give us a nice tree decomposition of the graph \(G\), say \(\mathcal{T}\).
We then pass \((G, \mathcal{T}, \mathbb{T})\) as inputs to the bounded treewidth algorithm described in Section \autoref{sec:vc-treewidth} and return \yes if and only if it return \yes.

Putting all these together, we get Algorithm \autoref{alg:fvc} that solves \fairvc. 

\begin{algorithm}
	\caption{Solving the \fairvc problem}
	\label{alg:fvc}
	\textit{Input}: A graph \(G\), a colouring function, \(c : V(G) \rightarrow [t]\) and tuple of integers \(\mathbb{T} = (k_1, k_2, \ldots , k_t)\) \\
	 \textit{Output}: Return \yes if \(G\) has a \(\mathbb{T}\)-fair vertex, otherwise return \no
	 \underline{\textbf{FVC(\(G, c, \mathbb{T}\))}}
	\begin{algorithmic}[1]
		\STATE Apply check \(1\) from the No-instance checks. If \no is returned then we return \no, otherwise we go the following line
		\STATE Apply check \(2\) from the No-instance checks. If \no is returned then we return \no, otherwise we go the following line
		\STATE Check if there is an \(i^{\star} \in [t]\), such that, \(c_{i^{\star}}(V(G)) = k_{i^{\star}}\). If there does exist such an \(i^{\star}\), then apply the reduction rule to return the instance, \((G - X,c',(k'_{1},\dotsc,k'_{i-1},k'_i,k'_{i+1},\dotsc,k'_{t}))\), where, $X = \{v \in V(G) | i^{\star} \in c(v)\}$, and for each $i \in [t]$, $k'_i = k_i - c_i(X)$. Otherwise, we go to the next line. 
		\IF{there are no vertices in \(G\) with degree at least \(3\)}
		\STATE Let components of \(G\) be \(C_{1},C_{2},\dotsc,C_{h}\)
		\FOR{\(j \in [h]\)}
		\STATE Make nice tree decomposition of \(C_j\), say \(T_j\)
		\ENDFOR
		\FOR{\(j \in [h - 1]\)}
		\STATE Add an edge between between \(T_j\) and \(T_{j + 1}\)
		\ENDFOR
		\STATE We get a nice tree decomposition of \(G\), say \(\mathcal{T}\)
		\STATE Pass \((G, \mathcal{T}, \mathbb{T})\) as inputs to the bounded treewidth algorithm and return whatever it returns
		\ELSE
		\STATE Let \(v\) be a vertex in \(G\) of degree at least \(3\). Let, \(H', H'', c', c'', \mathbb{T'}\) and \(\mathbb{T''}\) be defined as done in Section 3.1.3.
		\STATE Return \(FVC((H',c', \mathbb{T'})) \oplus FVC((H'',c'', \mathbb{T''}))\)
		\ENDIF
\end{algorithmic}
\end{algorithm}

\paragraph{Proof of correctness}

If for some $i \in [t]$, we have, $c_i(V(G)) < k_i$, then we cannot have a $(k_i)_{i = 1} ^t$-fair vertex cover of $G$, and thus we are dealing with a \noinstance.

If there exists a  pair \((i, j) \in [t] \times [t]\)
  such that (i) \(k_{i} = k_{j} = 0\) and (ii) there is an edge \((u,v) \in
  E(G)\) for which both \(i \in c(u)\) and \(j \in c(v)\) hold, then too we have a \noinstance, as we need to pick at least one of $u$ and $v$ in any vertex cover, but we cannot do so with the colour constraints.
  
  Suppose, there is an \(i^{\star} \in [t] \) such that \(c_{i^{\star}}(V(G)) =
  k_{i^{\star}}\). Let $X = \{v \in V(G) | i^{\star} \in c(v)\}$ and for each $i \in [t]$, $k'_i = k_i - c_i(X)$. Let $F^{\star}$ be a $(k'_i)_{i = 1} ^t$-fair vertex cover of $G - X$. Then $F^{\star} \cup X$ is a $(k_i)_{i = 1} ^t$-fair vertex cover of $G$. Conversely, let $F$ be a $(k_i)_{i = 1} ^t$-fair vertex cover of $G$. Then, $F \setminus X$ is a $(k'_i)_{i = 1} ^t$-fair vertex cover of $G - X$.

The correctness of the branching rule mainly depends on the fact that in a vertex cover of a graph, if any vertex of the said graph doesn't belong to the vertex cover, then the open neighbourhood of the graph must be in the vertex cover. This is exactly what the branching rule does, while maintaining the colour constraints.

Suppose, \((G,c,\ttuple)\) , is a \yesinstance. Then it has a $(k_i)_{i=1}^{t}$-fair vertex cover, say $F$. Let, $v \in V(G)$, such that, it has degree at least $3$. If $v \in F$, then, $F \setminus \{v\}$ is a $(k'_i)_{i = 1} ^t$-fair vertex cover of $H'$, where $H'$ and $(k'_i)_{i = 1} ^t$ are as defined in the branching rule above. Thus, \((H', c',
    \ttuple')\) is a \yesinstance. Otherwise, notice that, $N(v) \subseteq F$. Then, $F \setminus N(v)$ is a $(k''_i)_{i = 1} ^t$-fair VC of $H''$, where $H''$ and $(k''_i)_{i = 1} ^t$ are as defined in the branching rule above, and so \((H', c',
    \ttuple')\) is a \yesinstance.

Conversely, if \((H', c',
    \ttuple')\) is a \yesinstance is \yesinstance with $F'$ as the $(k'_i)_{i = 1} ^t$-fair vertex cover of $H'$, then $F' \cup \{v\}$ will be the $(k_i)_{i = 1} ^t$-fair vertex cover of $G$. Otherwise, if \((H'', c'',
    \ttuple'')\) is a \yesinstance is \yesinstance with $F''$ as the $(k''_i)_{i = 1} ^t$-fair vertex cover of $H''$, then $F'' \cup N(v)$ will be the $(k_i)_{i = 1} ^t$-fair vertex cover of $G$.
    
    We now turn to the base case, where $G$ is a graph with maximum degree $3$ and as such, it is a disjoint union of isolated vertices, paths and cycles. Let the connected
components of \(G\) be \(C_{1},C_{2},\dotsc,C_{h}\). Let \(\ttuple = (k_{1},
k_{2},\dots,k_{t})\). Then, we can get a nice tree decomposition of each component \(C_j\), say \(T_j\) in linear time. As adding an edge between two trees still gives us a tree, thus we can get a tree decomposition of the entire graph \(G\), say \(\mathcal{T}\). We pass on \((G, \mathcal{T}, \mathbb{T})\) as inputs to the treewidth based algorithm described in Section \autoref{sec:vc-treewidth}. Assuming that is correct, we have proved the correctness of our algorithm.

\paragraph{Analysis of the running time}
  
  The time taken for the checks in Section 3.1.1 and the reduction rule in Section 3.1.2 is $O(n^2 t)$.
  
  As this is a branching algorithm we will first try to get a bound on the number of nodes. We see that if we define $T(k)$ as follows:
  
  \begin{itemize}
      \item $T(k) = T(k - 1) + T(k - 3)$, when $k \geq 3$
      \item $T(k) = 1$, otherwise
  \end{itemize}

Then, if we substitute $\sum_{i = 1} ^t k_i$ for $k$, we will get an upper bound on the number of nodes of the search tree. By the claim in the book, $T(k) \leq 1.4656$, thus the total number of nodes in the search tree is at most $1.4656^k$. 

The time spend on the interior nodes, including the root one is $O(n^2 t)$. As for the leaves we will need \(O(n)\) time to form the tree decomposition of the graph \(G\), with an additional \(O(nt \prod_{i = 1} ^t k_i ^2)\) time for invoking the treewidth algorithm. Thus, the maximum time taken in any node is \\\(O(n^2 t \prod_{i = 1} ^t k_i ^2)\) and hence the total running time for the algorithm is \(O(n^2 t (\prod_{i = 1} ^t k_i ^2) 1.4656^k)\), where \(k = \sum_{i = 1} ^t k_i\).

\subsubsection{Polynomial kernel for a constant number of colours}\label{sec:kernel}

In this section we provide a kernelisation algorithm for the problem. A kernelisation algorithm for a parameterised problem is an algorithm which takes an instance, \((x,k)\) of the problem and in polynomial time return an equivalent instance, \((x',k')\) such that \(\lvert x' \rvert + k' \leq g(k)\), where \(g : \mathbb{N} \rightarrow \mathbb{N}\) is a computable function.

Given a graph, \(G\), the colouring function, \(c : V(G) \rightarrow (2^{[t]} \setminus \emptyset)\), a \(t\)-tuple of integers, \(\mathbb{T} = (k_1, \ldots, k_t)\), we have the following rules.

\paragraph{The Global \noinstance check}

We apply this rule once, right at the beginning.

\begin{description}
	\item[Check 0] If \(c_{i}(V(G)) < k_{i}\) holds for at least one \(i \in [t]\)
  then return \no. 

\end{description}

\paragraph{Setting aside isolated vertices}

If the above rule doesn't return \no, then we know that, for each \(i \in [t], c_i(V(G)) \geq k_i \). Let \(k_{max} = max\{k_i | i \in [t]\}\). We define a subset \(I*\) of \(V(G)\) as follows.
For every \(X \subseteq [t]\), if \(X \neq \emptyset\), we denote the set \(V_X := \{v \in V(G) | c(v) = X \land deg(v) = 0\}\). If \(\lvert V_X \rvert > k_{max}\), then keep any \(k_{max}\) of them in \(V_X\) and remove the rest. We define, \(I* := \bigcup_{X \subseteq [t] \land X \neq \emptyset} V_X\).

We then apply the following rules exhaustively.n order as they appear If at some point, any of then return \no, we terminate the algorithm and return \no for the original input.

\paragraph{The \noinstance check}\label{sec:nocheck1}

\begin{description}

\item[Check 1] If there is at least one pair \((i, j) \in [t] \times [t]\)
  such that (i) \(k_{i} = k_{j} = 0\) and (ii) there is an edge \((u,v) \in
  E(G)\) for which both \(i \in c(u)\) and \(j \in c(v)\) hold, then return \no.

\end{description}

\paragraph{The isolated vertex rule}

\begin{description}

	\item[Isolated vertex Rule] If \(v \in V(G)\), such that \(deg(v) = 0\) and \(v \notin I*\), then we return the instance, \((G - v, c', \mathbb{T})\), where \(c'\) is the restriction of \(c\) on \(V(G) \setminus \{v\}\).

\end{description}

\paragraph{The large neighbourhood rule}

\begin{description}
	\item[Large neighbourhood rule] If for some \(v \in V(G)\) such that for some \(i \in [t]\), \(\lvert N_i(v) \rvert > k_i\), then we return \((G - v, c', (k'_1, \ldots, k'_t))\), where \(c'\) is the restriction of \(c\) on \(V(G) \setminus \{v\}\) and for each \(i \in [t], k'_i = k_i - c_i(\{v\})\).

\end{description}

\paragraph{The final instance size}

\begin{description}
	\item[Return a kernel] Let \((G, c, \mathbb{T})\) be an input instance, such that the above rules are no more applicable. Then, if \(\lvert V(G) \rvert\ > (\sum_{i = 1} ^t k_i)^2 + \sum_{i = 1} ^t k_i \times (1 + 2^t)\) or \( \lvert E(G) \rvert >  (\sum_{i = 1} ^t k_i)^2\), we return \no, else we return this instance as our kernel.
\end{description}

Combining all the above rules, we have the following algorithm

\begin{algorithm}
	\caption{Kernel for the \fairvc problem}
	\label{alg:fvc-kernel}
	\textit{Input}: A graph \(G\), a colouring function, \(c : V(G) \rightarrow [t]\) and tuple of integers \(\mathbb{T} = (k_1, k_2, \ldots , k_t)\) \\
	 \textit{Output}: Return a kernel for the input instance
	\begin{algorithmic}[1]
		\STATE Apply check \(0\) from the No-instance checks. If \no is returned then we return \no, otherwise we go the following line
		\STATE \(k_{max} \leftarrow max\{k_i | i \in [t]\}\)
		\STATE \(I* \leftarrow \emptyset\)
		\FOR{\(X \in (2^{[t]} \setminus \emptyset)\)}
			\STATE \(V_X \leftarrow \emptyset\)
				\FOR{\(v \in V(G)\)}
				\IF{\(c(v) = X \land deg(v) = 0\)}
				\IF{\(\lvert V_X \rvert < k_{max}\)}
				\STATE \(V_X \leftarrow V_X \cup \{v\}\)
				\ENDIF
				\ENDIF
				\ENDFOR
		\STATE \(I* \leftarrow I* \cup V_X\)
		\ENDFOR
		\STATE We now call \autoref{alg:fvc-find-kernel} with the input \(FINDKERNEL(G, c, \mathbb{T})\)
\end{algorithmic}
\end{algorithm}

\begin{algorithm}
	\caption{Finding kernel function}
	\label{alg:fvc-find-kernel}
	\underline{\textbf{FINDKERNEL\((G, c, \mathbb{T})\)}}
\begin{algorithmic}[1]
	\STATE Apply Check 1 from \autoref{sec:nocheck1}. If it returns \no, then return \no, otherwise move the next line
	\IF{\(v \in V(G) \land deg(v) = 0 \land v \notin I*\)}
	\STATE Return \(FINDKERNEL(G - v, c', \mathbb{T})\), where \(c'\) is the restriction of \(c\) to \(V(G) \setminus \{v\}\).
	\ENDIF
	\FOR{\(v \in V(G)\)}
		\FOR{\(i \in [t]\)}
		\IF{\(\lvert N_i(v) \rvert > k_i\)}
		\STATE Return \(FINDKERNEL(G \mbox{-} v, c', (k'_1, \ldots , k'_t))\), where \(c'\) is the restriction of \(c\) on \(V(G) \setminus \{v\}\) and for each \(i \in [t], k'_i = k_i - c_i(\{v\})\)
		\ENDIF
		\ENDFOR
	\ENDFOR
	\IF{\((\lvert V(G) \rvert\ > (\sum_{i = 1} ^t k_i)^2 + \sum_{i = 1} ^t k_i \times (1 + 2^t)) \lor (\lvert E(G) \rvert >  (\sum_{i = 1} ^t k_i)^2)\)}
	\STATE Return \no
	\ELSE
	\STATE Return \((G, c, \mathbb{T})\)
	\ENDIF
 
\end{algorithmic}
\end{algorithm}

\paragraph{Safeness of the rules}

The correctness of Check 0 has already been proved in the branching algorithm case in \autoref{sec:fpt}. The reason for it being applied only once at the beginning, is due to the fact that, during the running of the rest of algorithm there could occur some deletions of vertices, which can result in temporary violations of colour constraint.

The correctness of the No-instance check has also been proven in \autoref{sec:fpt}.

We now come to the isolated vertex rule. Suppose, \((G, c, \mathbb{T})\) is a \yesinstance. If there exists \(v \in V(G)\) with \(deg(v) = 0\) and \(v \notin I*\), then we can safely delete \(v\). If, \(F\) is a \(\mathbb{T}\)-fair vertex cover of \(G\), which doesn't contain \(v\), then \(F\) is a \(\mathbb{T}\)-fair vertex cover of \(G - v\). In case, \(v \in F\), then notice that not all vertices from \(I*\) which have the same colour as \(v\) are not in \(F\). If that was the case, there would be more than \(k_{max}\) vertices in \(F\) with the same colour as \(v\), which would contradict that \(F\) is \(\mathbb{T}\)-fair. Thus, we have \(u \in I* \setminus F\) with \(c(u) = c(v)\). We then get, \((F \setminus \{v\}) \cup \{u\}\), a \(\mathbb{T}\)-fair vertex cover of \(G - v\). It's easy to see that \((G - v, c', \mathbb{T})\), where \(c'\) is the restriction of \(c\) to \(V(G) \setminus \{v\}\), being a \yesinstance implies that \((G, c, \mathbb{T})\) is a \yesinstance.

Let \((G, c, (k_i)_{i = 1} ^t)\) be an instance, with a vertex \(v\) of \(G\), such that there is an \(i^{\star} \in [t]\), with \(\lvert N_{i^{\star}}(v) \rvert > k_{i^{\star}}\). If \( \\ (G, c, (k_i)_{i = 1} ^t)\) is a \yesinstance, then there is a \((k_i)_{i = 1} ^t\)-fair vertex cover of \(G\), say \(F\). Then we must have \(v \in F\), otherwise we would have \(N(v) \subset F\) which would imply that, \(c_i(F) > k_i\) and contradict that \(F\) is \((k_i)_{i = 1} ^t\)-fair. We observe that, \(F \setminus \{v\}\) is a \((k'_i)_{i = 1} ^t\)-fair vertex cover of \(G - v\), where for each \(i \in [t]\), \(k'_i = k_i - c_i(\{v\})\). Conversely, if \((G - v, c', (k'_i)_{i = 1} ^t)\) is a \yesinstance, where \(c'\) is the restriction of \(c\) on \(V(G) \setminus \{v\}\), then we have a \((k'_i)_{i = 1} ^t\)-fair vertex cover of \(G - v\), say \(F'\). Then we observe that \(F' \cup \{v\}\) is a \((k_i)_{i = 1} ^t\)-fair vertex cover of \(G\).

We apply the large neighbourhood rule, only when the other rules are no more applicable. This gives us the following lemma.

\begin{lemma}
\label{lem:kernelsize}
	If \((G,c,\mathbb{T})\) is \yesinstance and none of the rules from the No-instance checks, the reduction rule, the isolated vertex rule and the large neighbourhood rule is applicable, then \(\lvert V(G) \rvert \leq (\sum_{i = 1} ^t k_i)^2 + \sum_{i = 1} ^t k_i \times (1 + 2^t)\) and \( \\ \lvert E(G) \rvert (\sum_{i = 1} ^t k_i)^2\).
\end{lemma}

\begin{proof}
	As the isolated vertex rule is no more applicable, thus \(G\) has \(\lvert I*\rvert\) isolated vertices and we can see that , \(\lvert I*\rvert \leq k_{max} 2^t \leq 2^t \sum_{i = 1} ^t k_i\). Since, we cannot apply the large neighbourhood rule anymore, thus any vertex of \(G\) will have have degree to be at most \(\sum_{i = 1} ^t k_i\). If
	\(F\) is \((k_i)_{i = 1} ^t\)-fair vertex cover of \(G\), then the number of vertices in \((G - I*) - F\) is at most \(\sum_{i = 1} ^t k_i \lvert F \rvert\). Thus, we get that \(\lvert V(G - I*) \rvert \leq \sum_{i = 1} ^t k_i \lvert F \rvert + \lvert F \rvert\) and hence, \(\lvert V(G) \rvert \leq \sum_{i = 1} ^t k_i \lvert F \rvert + \lvert F \rvert + \lvert I* \rvert \). Since, \(F\) is \((k_i)_{i = 1} ^t\)-fair, thus \(\lvert V(G) \rvert \leq (\sum_{i = 1} ^t k_i)^2 + \sum_{i = 1} ^t k_i \times (1 + 2^t)\). Also as every edge is covered by a vertex cover and every vertex has degree at most \(\sum_{i = 1} ^t k_i\), thus \(\lvert E(G) \rvert \leq (\sum_{i = 1} ^t k_i)^2\).
\end{proof}

\paragraph{Analysis of the running time}

\begin{lemma}
\label{lem:kerneltime}
	All the rules for kernelisation of the \fairvc problem can be applied in polynomial time.
\end{lemma}

\begin{proof}
	The time taken for the global \noinstance check is \(O(n t^2)\) and for constructing the set \(I*\) is \(O(n 2^t)\).

	To check if there exists a pair \((i,j) \in [t] \times [t]\) such that (i) \(k_{i} = k_{j} = 0\) and (ii) there is an edge \((u,v) \in
	E(G)\) for which both \(i \in c(u)\) and \(j \in c(v)\) hold requires \(O(n^2 t^2)\). 


	Applying the isolated vertex rule takes \(O(n)\) time.
	The time taken for applying the large neighbourhood rule takes \(O(n^2 t)\) time and the time to check the kernel size \(O(n^2)\).

	As each of the rules take polynomial time and we stop when the graph has been reduced to a certain size, thus we can apply the rules at most \(O(n)\) times and hence we have a polynomial time algorithm.

\end{proof}

\subsubsection{Parameterization by treewidth}\label{sec:vc-treewidth}

In the section, we design a fixed-parameter tractable algorithm for the  \fairvc problem, parameterized by the treewdith of the input graph. We assume that along with the graph \(G\), we are also given a tree decomposition \((\mathcal{T}, (B_x)_{x \in V(T)})\) of width \twidth. We have the followin theorem

\begin{theorem}
\label{thm:vctw}
	There is an algorithm, that given an $n$-vertex graph \(G\), a colouring function \(\fn{c}{V(G)}{2^{[t]} \setminus \set{\emptyset}}\), a \(t\)-tuple of non-negative integers \((k_i)_{i = 1}^{t}\), and a tree decomposition \((\mathcal{T}, (B_x)_{x \in V(\mathcal{T})})\) of \(G\) of width \(t_w\), runs in time \( \\ O(t l {t_w}^2 2^{t_w + 1} \prod_{i = 1} ^t k_i ^2)\) and decides correctly if \(G\) has a \( \\ (k_i)_{i=1}^{t}\)-fair vertex cover. 
\end{theorem}

We will solve the problem using the dynamic programming approach as done below.


\textbf{Definition of the states of the DP.}

For a node \(x\) of \(\mathcal{T}\), a subset \(S\) of \(B_x\), a \(t\)-tuple, \((r_i)_{i = 1} ^t\), where for each \(i \in [t]\), \(0 \leq r_i \leq k_i\), we define \(I_x[S,(r_i)_{i = 1} ^t]\) to be as follows:

\[
	I_x[S,(r_i)_{i = 1} ^t] = 1, \text{ if } G_x \footnotemark
	\text{ has an } (r_i)_{i = 1} ^t-fair \text{ vertex cover}
\]

	\footnotetext{Refer to \autoref{sec:preliminaries} for the definition} 
\(	~~~~~~~~~~~~~~~~~~~~~~~~~~~~~\text{which intersects } B_x \text{ at exactly } S \)

	\(I_x[S,(r_i)_{i = 1} ^t] = 0, \text{otherwise}\)


\textbf{Computing the entries of the DP table} We now fill the entries of the DP table as follows. Consider a node \(x \in V(\mathcal{T}) \), \(S \subseteq B_x\), and \((r_i)_{i = 1}^t\). At first we initialise all entries of the DP table with \(0\).

\begin{description}
	\item[Case 1:] \(x\) is a leaf node. Then \(B_x = \emptyset\), and hence \(S = \emptyset\). And we have

\[
	I_x[S,(r_i)_{i = 1} ^t] = \begin{cases}1, \text{ if for each } i \in [t], r_i = 0 \\0, \text{ otherwise} \end{cases}
\]

		For the following cases, we assume that all the DP entries have been filled for all the descendants of \(x\) in \(\mathcal{T}\).

\item[Case 2:] \(x\) is a forget node with child \(y\). Let \(v\) be the vertex forgotten at \(x\). That is, \(B_x = B_y \setminus \set{v}\) for some \(v \in B_y\). We then have

\[
	I_x[S,(r_i)_{i = 1} ^t] = I_y[S,(r_i)_{i = 1} ^t] \oplus I_x[S \cup \{v\},(r_i)_{i = 1} ^t] 
\]

\item[Case 3:] \(x\) is an introduce node with child \(y\). Let \(v\) be the vertex introduced at \(x\). That is, \(B_x = B_y \cup \set{v}\) for some \(v \notin B_y\). If \(v \notin S\), then we have

\[
	I_x[S,(r_i)_{i = 1} ^t] = I_y[S,(r_i)_{i = 1} ^t] 
\]

		Suppose \(v \in S\) and for some \(i^{\star} \in c(v)\), we have \(r_{i^{\star}} = 0\). Then we have

\[
	I_x[S,(r_i)_{i = 1} ^t] = 0
\]

In all other cases, we have

\[
	I_x[S,(r_i)_{i = 1} ^t] = I_y[S \setminus \{v\},(r'_i)_{i = 1} ^t], 
\]
		\(
	\text{where for each } i \in c(v), r'_i = r_i - 1,
	\text{and for each } i \notin c(v), r'_i = r_i
		\)

\item[Case 4:] \(x\) is an introduce edge node with child \(y\). Let \(uv\) be the edge introduced at \(x\). If \(S \cap \{u, v\} = \emptyset\), then we have

\[
	I_x[S,(r_i)_{i = 1} ^t] = 0
\]

Otherwise, we have

\[
	I_x[S,(r_i)_{i = 1} ^t] = I_y[S,(r_i)_{i = 1} ^t]
\]

\item[Case 5:] \(x\) is a join node with children, \(y\) and \(z\). We compute the value for \(I_x[S,(r_i)_{i = 1} ^t]\) by using the Algorithm \autoref{alg:Ix}, as shown below

\begin{algorithm}
\caption{Computing \(I_x[S,(r_i)_{i = 1} ^t]\)}
\label{alg:Ix}
  \begin{algorithmic}
    \FOR{\(a_{1} = 0\) to \(r_{1}\)}
    \STATE $\vdots$
    \FOR{\(a_{t} = 0\) to \(r_{t}\)}
	  \STATE \(I_{x}[S,{(r_{i})}_{i=1}^{t}] = (I_{x}[S,{(r_{i})}_{i = 1}^{t}] \oplus (I_{y}[S,{(a_{i})}_{i=1}^{t}] \odot I_{z}[S,{(r_{i} + c_i(S) - a_{i})}_{i=1}^{t}]))\)
    \ENDFOR
    \STATE $\vdots$
    \ENDFOR
  \end{algorithmic}
\end{algorithm}

\end{description}

\vspace{1mm}
Our algorithm will fill the DP tables for all choices of \(x, S\) and \((r_i)_{ i = 1} ^t\) using the rules above. And once we are done with that,we check \(I_x[\emptyset, (k_i)_{i = 1} ^t]\), where \(x\) is the root node of \(\mathcal{T}\). If it has been assigned the value \(1\), then we return \yes else we return \no.

\paragraph{Correctness of the Algorithm}

We have the following lemma

\begin{lemma}
\label{lem: fvctw}

	Given a node \(x\) of \(\mathcal{T}\), a subset \(S\) of \(B_x\) and a \(t\)-tuple \((r_i)_{i = 1} ^t\), our algorithm assigns the correct value to \(I_x[S, (r_i)_{i = 1} ^t]\), i.e., it assigns the value \(1\) to \(I_x[S, (r_i)_{i = 1} ^t]\) if and only if, the graph \(G_x\) has an \((r_i)_{i = 1} ^t\)-fair vertex cover such that it intersects \(B_x\) at exactly \(S\).

\end{lemma}

\begin{proof}
	We will prove this by induction.
	
	\textit{\textbf{Base case:}} \(x\) is a leaf node of \(\mathcal{T}\). Then \(B_x = \emptyset\) and hence \(S = \emptyset\). Also \(G_x\) is the empty graph and hence any vertex cover will be of size \(0\) which implies that it only has \((r_i)_{i = 1} ^t\)-fair vertex cover where for each \(i \in [t]\), \(r_i = 0\). Thus the assignment of our algorithm is correct.

	\textit{\textbf{Induction hypothesis:}} \(x\) is a node of \(\mathcal{T}\) such that we have correctly filled the DP entries for all nodes in the subtree of \(\mathcal{T}\), rooted at \(x\), including \(x\).

	\textit{\textbf{Inductive step:}} \(x\) is a non-leaf node of \(\mathcal{T}\) such that we have correctly filled the DP entries for all nodes in the subtree of \(\mathcal{T}\), rooted at \(x\). We then have the following cases.
	
\begin{description}
	\item[Case 1:] \(x\) is a forget vertex node with child \(y\) and it forgets the vertex \(v\).
		
		Let our algorithm assign the value \(1\) to \(I_x[S,(r_i)_{i = 1} ^t]\). This means at least one of \(I_y[S, (r_i)_{i = 1} ^t]\) and \(I_y[S \cup \{v\}, (r_i)_{i = 1} ^t]\) must be \(1\). By induction hypothesis, there exists an \((r_i)_{i = 1} ^t\)-fair vertex cover of \(G_y\), say \(F_y\) such that, \(F_y \cap B_y = S\) or \(F_y \cap B_y = S \cup \{v\}\). As \(G_x = G_y\), \(F_y\) is an \((r_i)_{i = 1} ^t\)-fair vertex cover of \(G_x\), such that \(F_y \cap B_x = S\).
		
		Conversely, suppose there exists an \((r_i)_{i = 1} ^t\)-fair vertex cover of \(G_x\), say \(F_x\) such that \(F_x \cap B_x = S\). Again as \(G_x = G_y\), then \(F_x\) is an \((r_i)_{i = 1} ^t\)-fair vertex cover of \(G_y\), such that \(F_x \cap B_y = S\) or \(F_x \cap B_y = S \cup \{v\}\). By induction hypothesis, \(I_y[S, (r_i)_{i = 1} ^t] = 1\) or \(I_y[S \cup \{v\}, (r_i)_{i = 1} ^t] = 1\). Thus, we have our algorithm assign \(1\) to \(I_x[S, (r_i)_{i = 1} ^t]\).
		
	\item[Case 2:] \(x\) is is an introduce vertex node with child \(y\) and it introduces the vertex \(v\). 

		Suppose, \(v \notin S\). Then any \((r_i)_{i = 1} ^t\)-fair VC of \(G_x\), say \(F_x\), such that \(F_x \cap B_x = S\), is also an \((r_i)_{i = 1} ^t\)-fair VC of \(G_y\), such that \(F_x \cap B_x = S\). Conversely, if we have an \((r_i)_{i = 1} ^t\)-fair VC of \(G_y\), say \(F_y\), such that \(F_y \cap B_y = S\), then as \(v\) is an isolated vertex in \(G_x\), we have an \((r_i)_{i = 1} ^t\)-fair VC of \(G_x\), such that \(F_y \cap B_x = S\). Thus, our algorithm correctly assigns the value \(I_y[S, (r_i)_{i = 1} ^t]\) to \(I_x[S, (r_i)_{i = 1} ^t]\).

		Suppose, \(v \in S\). Then any vertex cover of \(G_x\) which intersects \(B_x\) at \(S\) will have at least one vertex of colour \(i\), for each \(i \in c(v)\). Thus, our algorithm correctly assigns 

\[
	I_x[S, (r_i)_{i = 1} ^t] = 0, \text{ if for some } i^{\star} \in c(v), r_{i^{\star}} = 0
\]

		Now let's consider the case when \(v \in S\) and for each \(i \in c(v), r_i \geq 1\). Suppose our algorithm assigns the value \(1\) to \(I_x[S, (r_i)_{i = 1} ^t]\). Then it means that \(I_y[S \setminus \{v\}, (r'_i)_{i = 1} ^t] = 1\), where for each \(i \in c(v), r'_i = r_i - 1\) and for each \(i \notin c(v), r'_i = r_i\). By induction hypothesis, \(G_y\) has an \((r'_i)_{i = 1} ^t\)-fair vertex cover , say \(F_y\), such that \(F_y \cap B_y = S \setminus \{v\}\). Then, \(F_y \cup \{v\}\) is an \((r_i)_{i = 1} ^t\)-fair vertex cover of \(G_x\). Conversely, let \(F_x\) be an \((r_i)_{i = 1} ^t\)-fair vertex cover of \(G_x\), such that \(F_x \cap B_x = S\). Then, \(F_x \setminus \{v\}\) is an \((r'_i)_{i = 1} ^t\)-fair vertex cover of \(G_y\) such that, \(F_x \cap B_y = S \setminus \{v\}\). By induction hypothesis, \(I_y[S \setminus \{v\}, (r'_i)_{i = 1} ^t] = 1\) and thus our algorithm will assign \(1\) to \(I_x[S, (r_i)_{i = 1} ^t]\). 

	\item[Case 3:] \(x\) is is an introduce edge node with child \(y\) and it introduces the edge \(uv\). Suppose, \(S \cap \{u, v\} = \emptyset\). Then we cannot have a vertex cover of \(G_x\) which will intersect \(B_x\) at exactly \(S\) as then edge \(uv\) will not be covered. Thus, our algorithm correctly assigns the value \(0\) to \(I_x[S, (r_i)_{i = 1} ^t]\).

		Now let's consider the case, when \(S \cap \{u , v\} \neq \emptyset\). Let our algorithm assign the value \(1\) to \(I_x[S, (r_i)_{i = 1} ^t]\). This means that \(I_y[S, (r_i)_{i = 1} ^t] = 1\). By induction hypothesis, we will have an \((r_i)_{i = 1} ^t\)-fair vertex cover of \(G_y\), say \(F_y\), such that \(F_y \cap B_y = S\). We observe that any edge in \(G_x\), aside from \(uv\) is also an edge in \(G_y\) and hence it covered by \(F_y\). Also, as \(F_y \cap B_y = S\), then \(\{u, v\} \cap F_y \neq \emptyset\) and hence even the edge \(uv\) is also covered by \(F_y\). Thus, we get an \((r_i)_{i = 1} ^t\)-fair vertex cover of \(G_x\), \(F_y\), such that \(F_y \cap B_x = S\). Conversely, suppose \(F_x\) is an \((r_i)_{i = 1} ^t\)-fair vertex cover of \(G_x\), such that \(F_x \cap B_x = S\). As \(G_y\) is a subgraph of \(G_x\), with \(V(G_x) = V(G_y)\), then \(F_x\) is also an \((r_i)_{i = 1} ^t\)-fair vertex cover of \(G_y\), such that \(F_x \cap B_y = S\). By induction hypothesis, \(I_y[S, (r_i)_{i = 1} ^t] = 1\) and hence our algorithm will assign \(1\) to \(I_x[S, (r_i)_{i = 1} ^t] = 1\).

	\item[Case 4:] \(x\) is a join node with children \(y\) and \(z\).

		Suppose, our algorithm assigns \(1\) to \(I_x[S, (r_i)_{i = 1} ^t]\). Then we have \(t\)-tuple, \((a_i)_{i = 1} ^t\) such that, \(I_y[S, (a_i)_{i = 1} ^t] = 1\) and \(I_z[S, (r_i + c_i(S) - a_i)_{i = 1} ^t] = 1\). By induction hypothesis, \(G_y\) has an \((a_i)_{i = 1} ^t\)-fair vertex cover, say \(F_y\), such that \(F_y \cap B_y = S\) and \(G_z\) has an \((r_i + c_i(S) - a_i)_{i = 1} ^t\)-fair vertex cover, say \(F_z\), such that \(F_z \cap B_z = S\). Recall that, \(B_x = B_y = B_z\) and thus, \(F_y \cap B_y = F_z \cap B_z = S\). Then, \(F_y \cup F_z\) is an \((r_i)_{i = 1} ^t\)-fair vertex cover of \(G_x\) such that, \((F_y \cup F_z) \cap B_x = S\).

		Conversely, suppose \(F_x\) is an \((r_i)_{i = 1} ^t\)-fair vertex cover of \(G_x\) such that, \(F_x \cap B_x = S\). Let, for each \(i \in [t]\), we define \(a_i := c_i(F_x \cap V(G_y))\), which implies that \(0 \leq a_i \leq r_i\). Then, we have \(F_x \cap V(G_y)\), an \((a_i)_{i = 1} ^t\)-fair vertex cover of \(G_y\), such that \((F_x \cap V(G_y)) \cap B_y = S\), since \(B_x = B_y\). Also, we have \(F_x \cap V(G_z)\), an \((r_i + c_i(S) - a_i)_{i = 1} ^t\)-fair vertex cover of \(G_z\), such that \((F_x \cap V(G_z)) \cap B_z = S\), since \(B_x = B_z\). By induction hypothesis, \(I_y[S, (a_i)_{i = 1} ^t] = I_z[S, (r_i + c_i(S) - a_i)_{i = 1} ^t] = 1\). Thus our algorithm will assign the value \(1\) to \(I_x[S, (r_i)_{i = 1} ^t]\). Also, note that once have assigned the value \(1\) to \(I_x[S, (r_i)_{i = 1} ^t]\), we do not want to forget and get it overwritten by another choice of \((a_i)_{i = 1} ^t\)'s.

\end{description}

	From the above lemma, all the DP entries have been correctly evaluated by the algorithm and hence, when \(x\) is the root of \(\mathcal{T}\), then \(I_x[\emptyset, (k_i)_{i = 1} ^t] = 1\) if and only if \(G\) has a \((k_i)_{ i = 1} ^t\)-fair vertex cover.

\end{proof}

\paragraph{Analysis of running time}

 Let \(l := \lvert V(\mathcal{T}) \rvert \). For each node \(x\) of \(\mathcal{T}\), we have at most \(2^{t_w + 1}\) choices for \(S\) and for each \(i \in [t]\) we have \(k_i + 1\) choices for \(r_i\). Thus, we have at most \(l 2^{t_w + 1} \prod_{i = 1} ^t (k_i + 1)\) many cells in the DP table. We now evaluate the time taken to fill each cell, \(I_x[S, (r_i)_{i = 1}]\).

 To identify the leaf nodes we need \(O(l)\) time. Once done, we will need \(O(1)\) time to fill each entry for each leaf node.

Once, we are done with leaf nodes, we will need \(O({t_w}^2)\) time to identify the category of non-leaf node, say \(x\). If \(x\) is a forget vertex node, then we need \(O(1)\) time to fill each cell associated with \(x\).

If \(x\) is an introduce vertex node, which introduces the vertex \(v\), then we will need \(O(tw)\) time to determine if \(v\) is in \(S\) or not. If \(x \notin S\), then we need an additional \(O(1)\) time to enter the value of \(I_x[S, (r_i)_{i = 1}]\). If \(x \in S\), then we need to check which of the \(r_i\)'s are \(0\). This will be done in \(O(t)\) time and then we will need an additionl \(O(1)\) time to fill the entries.

If \(x\) is an introduce edge node, which introduces the edge \(uv\), then we need \(O(t_w)\) time to check if \(S \cap \{u, v\}\) or not. After that we need \(O(1)\) to fill each cell.

If \(x\) is a join node, then we need to spend \(O(\prod_{i = 1} ^t r_i)\) time to fill the cell \(I_x[S, (r_i)_{i = 1} ^t]\), which is at most \(O(\prod_{i = 1} ^t k_i)\).

Thus, the maximum time spend in a cell is \(O( {t_w}^2 t \prod_{i = 1} ^t k_i)\). And so the total time taken is \(O(t l {t_w}^2 2^{t_w + 1} \prod_{i = 1} ^t k_i ^2)\).

\section{\FVS}\label{sec:fvs}

In this section we study the \ffvsfull\ problem. Just like \fairvc, \ffvs\ is also \NPH. 
And we can show the \NP-hardness by a simple reduction from \FVS. Let \((G,k)\) be an instance of
\FVS. We assume without loss of generality that $G$ has at least $k$ vertices, for otherwise $(G, k)$ is a \yesinstance\ of \FVS, and in this case, we can return a trivial \yesinstance\ of \ffvs. Then \(G\) has an fvs of size at most \(k\) if and only \(G\) has an fvs of size exactly \(k\).
We now construct a \(1\)-coloured graph \((H,c)\) by setting (i) \(H=G\) and (ii)
\(c\) to be the function that assigns the colour set \(\{1\}\) to every vertex
of \(H\). Further, we set \(\ttuple=(k)\)
It is easy to see that \(S\) is an fvs of \(G\) of size exactly \(k\) if and only if \(S\) is a \ttuple-fair fvs of \(H\).

As we did for \fairvc, we study two different parameterizations of \ffvs---with respect to the treewidth of the input graph (see \autoref{sec:fvs-treewidth}) and with respect to the total colour budget (see \autoref{sec:fvs-fpt}). 
With respect to both the parameterizations, we design single-exponential FPT algorithms.

\subsubsection{Summary of Results From}~\cite{bodlaender2015deterministic}

\paragraph*{Terminology and notation for the rank-based approach.} We now introduce the following  operations on a collection of partitions of a set. For a set $U$, $\Pi(U)$ denotes the set of all partitions of $U$. 
For $v \in U$ and partition $p \in \Pi(U)$, we write $p - v$ to mean the partition of $U \setminus \set{v}$ obtained from $p$ by removing $v$. For two partitions $p, q \in \Pi(U)$, $p \sqcup q$ is a partition of $U$ defined as follows: Consider a graph $H$ with vertex $U$ and obtained by turning every set in $p \cup q$ into a clique. And we define $p \sqcup q$ to be the partition of $U$ into the  connected components of $H$. 
For a family of partitions $\mathcal{A} \subseteq \Pi(U)$, we define the following operations that deal with modifying $\mathcal{A}$ while adding an element to $U$, removing an element from $U$, merging a pair of sets in the partitions in $\mathcal{A}$, and joining $\mathcal{A}$ with another family $\mathcal{B} \subseteq \Pi$. 

Formally, for $\mathcal{A} \subseteq \Pi$ we define:

\begin{itemize}
    \item {\bf Insert.} For $v \notin U$, $\ins(v, \mathcal{A}) = \{ p \cup \set{\set{v}} ~|~ p \in \mathcal{A}\}$. That is, $\ins(v, \mathcal{A})$ adds $v$ to every partition in $\mathcal{A}$ as a singleton set, and thus $\ins(v, \mathcal{A}) \subseteq \Pi(U \cup \set{v})$. 
    
    \item {\bf Project.} For $v \in U$, $\proj(v, \mathcal{A}) = \{p - v ~ | ~ p \in \mathcal{A} \text{ and } \set{v} \notin p\}$. That is, $\proj(v, \mathcal{A})$ removes $v$ from every partition in $\mathcal{A}$ and discards those partitions in which $v$ was present as a singleton set. Thus, $\proj(v, \mathcal{A}) \subseteq \Pi(U \setminus \set{v})$. 
    
    \item {\bf Glue.} For $u, v$, let $\hat U = U \cup \set{u, v}$. The elements $u$ and $v$ may or may not be present in $U$. Let $\hat U[\set{u, v}]$ be the partition of $\hat U$ that contains the set $\set{u, v}$ and all elements of $\hat U \setminus \set{u, v}$ as singleton sets. For $p \in \mathcal{A}$, let $\hat p = p$ if $u, v \in U$; $\hat p = p \cup \set{\set{u}}$ if $v \in U$ and $u \notin U$; $\hat p = p \cup \set{\set{v}}$ if $u \in U$ and $v \notin U$; and $\hat p = p \cup \set{\set{u}} \cup \set{\set{v}} $ if $u, v \notin U$. Now, 
    \[
    \glue(uv, \mathcal{A}) = \{\hat U[\set{u, v}] \sqcup (p \cup \set{\set{u}, \set{v}}) ~ | ~ p \in \mathcal{A} \}.
    \]
    That is, $\glue(uv, \mathcal{A}) \subseteq \Pi(\hat U)$ adds the elements $u$ and $v$ to $\hat U$ (if they are not already present in $U$) and merges the blocks of $\hat p$ that contain $u$ and $v$. Thus, $\glue(uv, \mathcal{A}) \subseteq \Pi(\hat U)$.  
    
    \item {\bf Join.} For $\mathcal{B}\subseteq \Pi(U)$, $\join(\mathcal{A}, \mathcal{B}) = \{ p \sqcup q ~|~ p \in \mathcal{A} \text{ and } q \in \mathcal{B} \}$. 
\end{itemize}

These operations were defined by Bodlaender et al.~\cite{bodlaender2015deterministic}. But their definitions were more general, in the sense that they were working with families of weighted partitions; that is, families of the type $\mathcal{A} \subseteq \Pi(U) \times \mathbb{N}$. But weights are not necessary for our purposes. 

\begin{proposition}\cite{bodlaender2015deterministic}
\label{prop:operations}
We can perform each of the operations insert, project, glue in time $S \times |U|^{\cO(1)}$, where $S$ is the size of the input of the operation. Given, $\mathcal{A}, \mathcal{B}$, we can compute $\join(\mathcal{A}, \mathcal{B})$ in time $\card{A} \cdot \card{B} \cdot \card{U}^{\cO(1)}$. 
\end{proposition}

\begin{proposition}\cite{bodlaender2015deterministic}
\label{prop:reduce}
There exists an algorithm that given a family of partitions $\mathcal{A} \subseteq \Pi(U)$ as input, runs in time $2^{\cO(U)} \cdot \card{\mathcal{A}}^{\cO(1)}$, and outputs a family of partitions $\mathcal{A}' \subseteq \mathcal{A}$ such that (i) $\card{\mathcal{A}'} \leq 2^{\card{U} - 1}$ and (ii) for every $p \in \mathcal{A}$ and $q \in \Pi(U)$ with $p \sqcup q = \set{U}$, there exists $p' \in \mathcal{A}'$ such that $p' \sqcup q = \set{U}$. 
\end{proposition}

\subsubsection{\ffvsfull\ on Graphs of Bounded Treewidth} 
\label{sec:fvs-treewidth}
In the section, we design a fixed-parameter tractable algorithm for the  \ffvsfull\ problem, parameterized by the treewdith of the input graph. As in the previous section, we assume that along with the graph $G$, we are also given a tree decomposition $(\mathcal{T}, (B_x)_{x \in V(\mathcal{T})})$ of width \twidth. Specifically, we prove the following theorem. 

\begin{theorem}
\label{thm:fvs}
	There is an algorithm, that given an $n$-vertex graph $G$, a colouring function $\fn{c}{V(G)}{2^{[t]} \setminus \set{\emptyset}}$, a $t$-tuple of non-negative integers $(k_i)_{i = 1}^{t}$, and a tree decomposition $(\mathcal{T}, (B_x)_{x \in V(\mathcal{T})})$ of $G$ of width \twidth, runs in time $n^{\cO(1)} 2^{\cO({\twidth})}$, and decides correctly if $G$ has a $(k_i)_{i=1}^{t}$-fair fvs. 
\end{theorem}

Notice that our goal is to solve the \ffvs\ problem in time \fvsruntime. We can solve the problem in time $\twidth^{\twidth} n^{\cO(1)}$ using a straightforward dynamic programming algorithm. But to achieve the runtime of \fvsruntime, we use the rank-based approach introduced by Bodlaender et al.~\cite{bodlaender2015deterministic}. And this requires us to consider an alternative formulation of the problem. Notice that an $n$-vertex graph $G$ has a feedback vertex of size at most $k$ if and only if there exists $X \subseteq V(G)$ such that $\card{X} \geq n - k$ and $G[X]$ is acyclic (and hence $G[X]$ has at most $\card{X} - 1$ edges). Thus, instead of looking for a feedback vertex set of size at most $k$, we may as well look for an acyclic subgraph of $G$ of size at least $n - k$. Now, to ensure that $G[X]$ is acyclic, it is sufficient to ensure that $G[X]$ is connected and has exactly $\card{X} - 1$ edges. (Recall that an $n$-vertex connected graph is acyclic if and only if it has exactly $n - 1$ edges). But in general, the acyclic subgraph obtained by deleting a feedback vertex set need not be connected. So to ensure that we are left with an acyclic, connected subgraph (i.e., a tree) after deleting the feedback vertex set, we introduce a new vertex $v_0$ and make it adjacent to all the other vertices. Let $E_0$ be the set of new edges added to the graph this way, i.e., the edges incident with $v_0$. We can then look for a pair $(X, X_0)$, where $X \subseteq V(G) \cup \set{v_0}$, $X_0 \subseteq E_0$ and the subgraph with vertex set $X$ and edge set $E[X \setminus \set{v_0}] \cup X_0$ is a tree. (The tree thus contains all the edges with both the endpoints in $X \setminus \set{v_0}$ and some of the edges incident with $v_0$.) We, of course, have to adapt this formulation for our setting to satisfy the fairness constraints. To that end, we first formally prove the following lemma. 

\begin{lemma}
\label{lem:v0}
Let $(G, c, (k_i)_{i = 1}^{t})$ be an instance of \ffvs. Let $n_i = c_i(V(G))$. Consider the graph $G'$ defined by $V(G') = V(G) \cup \set{v_0}$ and $E(G') = E(G) \cup E_0$, where $v_0 \notin V(G)$ and $E_0 = \{ vv_0 ~|~ v \in V(G) \}$. Let $\fn{c'}{V(G')}{2^{[t]} \setminus \set{\emptyset}}$ be the extension of $c$ defined as follows: $c'(v_0) = [t]$ and $c'(v) = c(v)$ for every $v \in V(G)$. Then $G$ has a $(k_i)_{i = 1}^t$-fair feedback vertex set if and only if there exists a pair $(X, X_0)$ such that (a) $X \subseteq V(G) \cup \set{v_0}$ with $v_0 \in X$ and $X_0 \subseteq E_0$, (b) the subgraph of $G'$ with vertex set $X$ and edge set $E_{G'}[X \setminus \set{v_0}] \cup E_0$ is connected, (c) $\card{E_{G'}[X \setminus \set{v_0}] \cup X_0} = \card{X} - 1$ and (d) $c'_i(X) = n_i - k_i + 1$ for every $i \in [t]$. 
\end{lemma}

\begin{proof}
Assume that $G$ has a $(k_i)_{i = 1}^t$-fair feedback vertex set, say $F \subseteq V(G)$. We will show that there exists a pair $(X, X_0)$ with the required properties. First, let $X = \set{v_0} \cup (V(G) \setminus F)$. 
Then for every $i \in [t]$, since $c'(v_0) = [t]$ and $c'_i(V(G) \setminus F) = c_i(V(G) \setminus F) = n_i - k_i$, we have $c'_i(X) = n_i - k_i + 1$. Now, let $C_1, C_2,\ldots, C_{\ell}$ be the connected components of $G - F$. For each $j \in \ell$, fix a vertex $v_i$ in $C_i$. Let $X_0 = \set{v_0 v_j ~|~ j \in [\ell]} \subseteq E_0$. Then the subgraph with vertex set $X$ and edge set $E_{G'}[X \setminus \set{v_0}] \cup X_0 = E(G - F) \cup X_0$ is connected. Observe now that adding the vertex $v_0$ and the edges in $X_0$ to the (acyclic) graph $G - F$ does not create any cycles, as $v_0$ has exactly one neighbour in each connected component of $G - F$. Thus the subgraph with vertex set $X$ and edge set $E_{G'}[X \setminus \set{v_0}] \cup X_0 = E(G - F) \cup X_0$ is both connected and acyclic, and hence $\card{E_{G'}[X \setminus \set{v_0}] \cup X_0} = \card{X} - 1$. We have thus shown that properties (a)-(d) in the lemma statement hold for our choice of $(X, X_0)$.

Conversely, assume that there exists a pair $(X, X_0)$ that satisfy properties (a)-(d). Let $F' = V(G) \setminus X$ and let $H(X, X_0)$ be the subgraph of $G'$ vertex set $X$ and edge set $E_{G'}[X \setminus \set{v_0}] \cup X_0$.  Since $H(X, X')$ is connected and has exactly $\card{X} - 1$ edges, it is acyclic. Notice that $G - F' = G[X \setminus \set{v_0}]$ is a subgraph of $H(X, X_0)$, and hence $G - F'$ is acyclic. Now, consider $i \in [t]$. Since $c'(v_0) = [t]$ and $c'_i(X) = n_i - k_i + 1$, we have $c'_i(X \setminus \set{v_0}) = n_i - k_i$, which implies that $c_i(X \setminus \set{v_0}) = n _i - k_i$. Hence $c_i(F') = k_i$ as $F' = V(G) \setminus X$. Thus $F$ is a $(k_i)_{i=1}^t$-fair fvs of $G$. 
-\end{proof}

Let $G$ and $G'$ be as in Lemma~\ref{lem:v0}. 
Recall that we have a nice tree decomposition $(\mathcal{T}, (B_x)_{x \in V(\mathcal{T})})$ of $G$. We add the vertex $v_0$ to every bag $B_x$; and make the decomposition nice again, which is a nice tree decomposition of $G'$.  Let us denote the resulting nice tree decomposition by $(\mathcal{T}', (B_x)_{x \in V(\mathcal{T}')})$. 
For a node $x \in V(\mathcal{T}')$, recall that the graph $G'_x$ with $V(G'_x) = V_x$ and $E(G'_x) = E_x$ is the subgraph of $G'$ that consists of all the vertices and edges introduced in the subtree of $\mathcal{T}$ rooted at $x$. Consider $x \in V(\mathcal{T}')$. For $X \subseteq V_x$ and $X_0 \subseteq E_0 \cap E_x$, let $H_x(X, X_0)$ be the subgraph of $G'_x$ with vertex set $X$ and edge set $E_{G'_x}[X \setminus \set{v_0}] \cup X_0$. 

We are now ready to design our dynamic programming algorithm, which we call \fvsalgo. We first define the states of our DP. 


\textbf{Definition of the states of the DP.} For every node $x$ of $\mathcal{T}'$, for every $t$-tuple, $(r_i)_{i = 1} ^t$, where for each $i \in [t]$, $0 \leq r_i \leq n_i$, for integers, $j_1$, $j_2$, where $0 \leq j_1 \leq \card{V(G')}$ and $0 \leq j_2 \leq \card {E(G')}$, and for each $S \subseteq {B_x}$ we define

\(
A_x(S, j_1, j_2, (r_i)_{i = 1} ^t) := \{ p ~|~ p \in \Pi(S)
\\
~~~~~~~~~~~~~~~~~~~~~~~~~~~~~~~~~~~~~~~~~~~~~~~~~~~~~~\text{ and } \mathcal{E}_x(p, S, j_1, j_2, (r_i)_{i = 1} ^t) \neq \emptyset \},
\\
\)
where
$\mathcal{E}_x(p, S, j_1, j_2, (r_i)_{i = 1} ^t)$ is the set of all ordered pairs $(X, X_0)$ such that 
\begin{enumerate}[label=(\roman*)]
\item $X \subseteq V_x$ and $X_0 \subseteq E_0 \cap E_x$;
\item $\card{X} = j_1$ and $\card{E_{G_x}[X \setminus \set{v_0}] \cup X_0} = j_2$, (i.e., the graph $H_x(X, X_0)$ has $j_1$ vertices and and $j_2$ edges); 
\item $c_i(X) = r_i$ for every $i \in [t]$; 
\item $X \cap B_x = S$;
\item $v_0 \in S$ if $v_0 \in B_x$;
\item for every $u \in X \setminus B_x$, there exists $u' \in S$ such that $u$ and $u'$ are connected in the graph $H_x(X, X_0)$; and
\item for every $v_1, v_2 \in S$, $v_1$ and $v_2$ are in the same block in the partition $p$ if and only if $v_1$ and $v_2$ are connected in the graph $H_x(X, X_0)$.  
\end{enumerate}

Before moving to the computation of the entries of the DP table, let us complete the description of our algorithm and prove its correctness. 

\textbf{Description of \fvsalgo.} Given an instance $(G, c, (k_i)_{i = 1}^t)$ of \ffvs, along with a nice tree decomposition $(\mathcal{T}, (B_x)_{x \in V(\mathcal{T})})$, we first construct $G'$, $c'$ and $(\mathcal{T}', (B_x)_{x \in V(\mathcal{T}')})$ as described above. Then we compute the entries $A_x(S, j_1, j_2, (r_i)_{i = 1} ^t)$ for all choices of $x, S, j_1, j_2$ and $(r_i)_{i = 1}^t$. (Assume for now that we can compute all the table entries correctly.) Let $\hat x$ be the root of $\mathcal{T}'$ and let $\hat y$ be the unique child of $\hat x$. If $A_{\hat y}(\set{v_0}, j_1, j_1 - 1, (n_i - k_i + 1)_{i = 1} ^t)$ is non-empty for some $j_1$, then we return that our original instance $(G, c, (k_i)_{i = 1}^t)$ of \ffvs\ is a yes-instance, and otherwise we return that it is a no-instance. This completes the description of \fvsalgo. 

\begin{lemma}
\label{lem:fvs_correctness}
\fvsalgo\ is correct.
\end{lemma}

\begin{proof}
To prove the correctness of \fvsalgo, assume first that $A_{\hat y}(\set{v_0}, j_1, j_1 - 1, (n_i - k_i + 1)_{i = 1} ^t)$ is non-empty. 
Recall that $v_0 \in B_x$ for every non-leaf, non-root node $x$ of $\mathcal{T}'$. 
That is, we forget the vertex $v_0$ at the root $\hat x$. 
Since $\hat y$ is the unique child of $\hat x$ and $B_{\hat x} = \emptyset$, we have $B_{\hat y} = \set{v_0}$. 
Hence $\Pi(\set{v_0}) = \set{v_0}$. 
Now, since $A_{\hat y}(\set{v_0}, j_1, j_1 - 1, (n_i - k_i + 1)_{i = 1} ^t) \neq \emptyset$, we have $\set{v_0} \in A_{\hat y}(\set{v_0}, j_1, j_1 - 1, (n_i - k_i + 1)_{i = 1} ^t)$ and  $\mathcal{E}_{\hat y}(\set{v_0}, j_1, j_1 - 1, (n_i - k_i + 1)_{i = 1} ^t) \neq \emptyset$. 
Let $(X, X_0) \in \mathcal{E}_{\hat y}(\set{v_0}, j_1, j_1 - 1, (n_i - k_i + 1)_{i = 1} ^t)$. 
Consider the graph $H_{\hat x}(X, X_0)$. 
By condition (ii) in the definition of $\mathcal{E}_{\hat y}(\emptyset, j_1, j_1 - 1, (n_i - k_i + 1)_{i = 1} ^t)$, the graph  $H_{\hat y}(X, X_0)$ has $\card{X} = j_1$ vertices and $j_1 - 1$ edges; by condition (iii) $c_i(X) = n_i - k_i + 1$ for every $i \in [t]$; and by condition (vi), every vertex $u \in X \setminus \set{v_0}$ is connected to $v_0$ in the graph $H_{\hat y}(X, X_0)$, which implies that $H_{\hat y}(X, X_0)$ is connected. 
Therefore, by Lemma~\ref{lem:v0}, $G$ has a $(k_i)_{i = 1}^t$-fair fvs. 

Conversely, assume that $G$ has a $(k_i)_{i = 1}^t$- fair fvs. 
Then by Lemma~\ref{lem:v0}, there exists a pair $(X, X_0)$ such that (a) $X \subseteq V(G) \cup \set{v_0}$ with $v_0 \in X$ and $X_0 \subseteq E_0$, (b) the subgraph of $G'$ with vertex set $X$ and edge set $E[X \setminus \set{v_0}] \cup E_0$ is connected, (c) $\card{E[X \setminus \set{v_0}] \cup X_0} = \card{X} - 1$ and (d) $c'_i(X) = n_i - k_i + 1$ for every $i \in [t]$. Let $j_1 = \card{X}$. Conditions (a)-(d) together imply that $(X, X_0) \in \mathcal{E}_{\hat y}(\set{v_0}, j_1, j_1 - 1, (n_i - k_i + 1)_{i = 1} ^t)$. And since we have $\set{v_0} \in \Pi(\set{v_0})$, we can conclude that $A_{\hat y}(\set{v_0}, j_1, j_1 - 1, (n_i - k_i + 1)_{i = 1} ^t) \neq \emptyset$. 

\end{proof}

\textbf{Computation of the entries of the DP table. } We now fill the entries of the DP table as follows. Consider a node $x \in V(\mathcal{T}') $, $S \subseteq B_x$, $j_1, j_2$ and $(r_i)_{i = 1}^t$. 

\begin{description}
\item[Case 1:] $x$ is a leaf node. Then $B_x = \emptyset$, and hence $S = \emptyset$. And we have,

\(
		A_x(\emptyset, j_1, j_2, (r_i)_{i = 1} ^t) = \{ \emptyset \}, \text{ if }  j_1 = j_2 = r_1 = \ldots = r_t = 0
\)

and

\(
		A_x(\emptyset, j_1, j_2, (r_i)_{i = 1} ^t) =  \emptyset , \text{ otherwise}
\)

\item[Case 2:] $x$ is a forget node with child $y$. Let $v$ be the vertex forgotten at $x$. That is, $B_x = B_y \setminus \set{v}$ for some $v \in B_y$. We then have
\[
A_x(S, j_1, j_2, (r_i)_{i = 1} ^t) = A_y(S, j_1, j_2, (r_i)_{i = 1} ^t) \cup 
\]

		\(\proj(v, A_y(S \cup \set{v}, j_1, j_2, (r_i)_{i = 1} ^t)). \)


\item[Case 3:] $x$ is an introduce vertex node with child $y$. Let $v$ be the vertex introduced at $x$. That is, $B_x = B_y \cup \set{v}$ for some $v \notin B_y$. 
For $i \in [t]$, let $r'_i = r_i - 1$ if $i \in c'(v)$ and $r'_i = r_i$ if $i \notin c'(v)$. We then have 
\[
A_x(S, j_1, j_2, (r_i)_{i = 1} ^t) = \begin{cases} \emptyset, \text{ if } v = v_0 \text{ and } v \notin S, \\
 									\ins(v, A_y(S \setminus \set{v}, j_1 - 1, j_2, (r'_i)_{i = 1} ^t)), \text{ if } v \in S, \\ 
									A_y(S, j_1, j_2, (r_i)_{i = 1} ^t), \text{ otherwise.}
						\end{cases}
\]

\item[Case 4:] $x$ is an introduce edge node with child $y$. Let $e = uv$ be the edge introduced at $x$. We have as follows.

		
		When, \(u = v_0\) and \(v \in S\),
\[
	A_x(S, j_1, j_2, (r_i)_{i = 1} ^t) = A_y(S, j_1, j_2, (r_i)_{i = 1} ^t) \cup 
\]

		\(\glue(v_0v, A_y(S, j_1, j_2 - 1, (r_i)_{i = 1} ^t))\).

		When, \(v = v_0\) and \(u \in S\),
\[
	A_x(S, j_1, j_2, (r_i)_{i = 1} ^t) = A_y(S, j_1, j_2, (r_i)_{i = 1} ^t) \cup 
\]

		\(\glue(v_0u, A_y(S, j_1, j_2 - 1, (r_i)_{i = 1} ^t))\).

		When, both \(u\) and \(v\) are in \(S\),
\[
	A_x(S, j_1, j_2, (r_i)_{i = 1} ^t) = \glue(uv,  A_y(S, j_1, j_2 - 1, (r_i)_{i = 1} ^t))
\]

And for all other cases,
\[
	A_x(S, j_1, j_2, (r_i)_{i = 1} ^t) = A_y(S, j_1, j_2, (r_i)_{i = 1} ^t)
\]

\item[Case 5:] $x$ is a join node with children $y$ and $z$. We have
\[
	A_x(S, j_1, j_2, (r_i)_{i = 1}^t) = \bigcup_{\substack{j^y_1, j^y_2; ~ j^z_1, j^z_2; ~ (r^y_i, r^z_i)_{i=1}^{t} \\ j^y_1 + j^z_1 = j_1 + \card{S} \\ j^y_2 + j^z_2 = j_2 \\ (r^y_i + r^z_i = r_i + c'_i(S))_{i = 1}^t}} \join(\mathcal{Y}, \mathcal{Z})
\]

where, \(\mathcal{Y} = A_y(S, j^y_1, j^y_2, (r^y_i)_{i = 1} ^t)\) and \( \\ \mathcal{Z} = A_z(S, j^z_1, j^z_2, (r^z_i)_{i = 1} ^t)\).
\end{description}

This completes the description of the algorithm \fvsalgo. We now show that our computation of the table entries is correct.
\begin{lemma}
\label{lem:fvs_dp_correctness}
For every node $x \in V(\mathcal{T})'$, $S \subseteq B_x$, $j_1, j_2, (r_i)_{i = 1}^t$, \fvsalgo\ computes the entry $ \\ A_x(S, j_1, j_2, (r_i)_{i = 1}^t)$ correctly. 
\end{lemma}

\begin{sloppypar}
\begin{proof}
Consider $x \in V(\mathcal{T})'$, $S \subseteq B_x$, $j_1, j_2, (r_i)_{i = 1}^t$, where $x$ is not the root of $\mathcal{T}'$. We prove the lemma by induction on the tree $\mathcal{T}$. 
That is, assuming that all the entries corresponding to every descendant of $x$ is computed correctly, we prove that  $A_x(S, j_1, j_2, (r_i)_{i = 1}^t)$ is also computed correctly. 

Suppose first that $x$ is a leaf node. 
Then $B_x = \emptyset$ and the only choice for $S \subseteq B_x$ is $S = \emptyset$. 
And we have $\pi(S) = \set{\emptyset}$. 
By definition, $(\emptyset, \emptyset) \in \mathcal{E}_x(\set{\emptyset}, \emptyset, j_1, j_2, (r_i)_{i = 1} ^t)$ if $j_1 = j_2 =  r_1 = \ldots = r_t = 0$; and $\mathcal{E}_x(\set{\emptyset}, \emptyset, j_1, j_2, (r_i)_{i = 1} ^t) = \emptyset$ otherwise. 
We thus correctly set $A_x(S, j_1, j_2, (r_i)_{i = 1}^t) = \set{\emptyset}$ when $j_1 = j_2 =  r_1 = \ldots = r_t = 0$, and $A_x(S, j_1, j_2, (r_i)_{i = 1}^t) = \emptyset$ otherwise. 

Assume now that $x$ is a non-leaf node, and that we have correctly computed all the entries corresponding to every descendant of $x$. 


\textbf{Case 1. } Suppose that $x$ is a forget node with child $y$. Let $x$ forget the vertex $v$. 
Now, to prove that our computation is correct, we must prove that for a partition $p \in \Pi(S)$, $p \in A_x(S, j_1, j_2, (r_i)_{i = 1} ^t)$ if and only if $p \in A_y(S, j_1, j_2, (r_i)_{i = 1} ^t) \cup \proj(v, A_y(S \cup \set{v}, j_1, j_2, (r_i)_{i = 1} ^t))$. 

Suppose first that $p \in A_x(S, j_1, j_2, (r_i)_{i = 1} ^t)$. 
Then $\mathcal{E}_x(p, S, j_1, j_2, (r_i)_{i = 1} ^t) \neq  \emptyset $. 
Consider  $(X, X_0) \in \mathcal{E}_x(p, S, j_1, j_2, (r_i)_{i = 1} ^t)$. 
There are two possibilities: either $v \notin X$ or $v \in X$. 
Suppose $v \notin X$. 
Observe then that $(X, X_0) \in \mathcal{E}_y((p, S, j_1, j_2, (r_i)_{i = 1} ^t))$, and hence $p \in A_y(S, j_1, j_2, (r_i)_{i = 1}^t)$. 
Suppose now that $v \in X$. 
Then, since $v \notin B_x$, condition (vi) in the definition of $\mathcal{E}_x(p, S, j_1, j_2, (r_i)_{i = 1} ^t)$ implies that there exists $u' \in B_x$ such that $v$ and $u'$ are in the same connected component of the graph $H_x(X, X_0)$. 
Let $p'$ be the partition of $S \cup \set{v}$ such that for $w, w' \in S \cup \set{v}$, $w$ and $w'$ are in the same part of $p$ if and only if they are in the same connected component of $H_x(X, X_0)$. 
Notice then that $p = p' - v$ and that $u'$ and $v$ belong to the same part of $p'$. That is, $v'$ is not present as a singleton set in the partition $p'$. 
Also, note that $(X, X_0) \in \mathcal{E}_y((p', S \cup \set{v} , j_1, j_2, (r_i)_{i = 1} ^t))$, and hence $p' \in A_y(S \cup \set{v}, j_1, j_2,  (r_i)_{i = 1} ^t)$. Hence, $p \in \proj(v, A_y(S \cup \set{v}, j_1, j_2,  (r_i)_{i = 1} ^t)$. 

Conversely, suppose now that $p \in A_y(S, j_1, j_2, (r_i)_{i = 1} ^t) \cup \proj(v, A_y(S \cup \set{v}, j_1, j_2, (r_i)_{i = 1} ^t))$. 
First, consider the case when $p \in A_y(S, j_1, j_2, (r_i)_{i = 1} ^t)$. Then $\mathcal{E}_y(p, S, j_1, j_2,  (r_i)_{i = 1} ^t) \neq \emptyset$. 
Consider $(X, X_0) \in \mathcal{E}_y(p, S, j_1, j_2,  (r_i)_{i = 1} ^t)$. 
Then $X \cap B_y = S$, which implies that $v \notin X$. Observe then that $(X, X_0) \in \mathcal{E}_x(p, S, j_1, j_2,  (r_i)_{i = 1} ^t)$ as well, which implies that $p \in A_x(S, j_1, j_2, (r_i)_{i = 1} ^t)$. 
Now, suppose that $p \in \proj(v, A_y(S \cup \set{v}, j_1, j_2,  (r_i)_{i = 1} ^t)$. Let $q \in A_y(S \cup \set{v}, j_1, j_2,  (r_i)_{i = 1} ^t$ be such that $p = q - v$ and $\set{v} \notin q$.  
Such a partition $q \in \Pi(S \cup \set{v})$ exists as $p \in \proj(v, A_y(S \cup \set{v}, j_1, j_2,  (r_i)_{i = 1} ^t)$. 
Again, since $q \in A_y(S \cup \set{v}, j_1, j_2,  (r_i)_{i = 1} ^t$, there exists $(X, X_0) \in  \mathcal{E}_y(q, S \cup \set{v}, j_1, j_2,  (r_i)_{i = 1} ^t)$. 
We now claim that $(X, X_0) \in  \mathcal{E}_x(p, S, j_1, j_2,  (r_i)_{i = 1} ^t)$ as well, which would imply that $p \in A_x(S, j_1, j_2, (r_i)_{i = 1} ^t)$. 
To see this, notice first that since $x$ is a forget node, we have $G'_x = G'_y$ and hence $H_x(X, X_0) = H_y(X, X_0)$. Since $p = q - v$, for vertices $v_1, v_2 \in S$, $v_1$ and $v_2$ are in the same block of $p$ if and only they are in the same block of $q$ and hence in the same connected component of $H_y(X, X_0) = H_x(X, X_0)$. That is, $(X, X_0)$  satisfies condition (vii) in the definition of $\mathcal{E}_x(p, S, j_1, j_2,  (r_i)_{i = 1} ^t)$. Now, since $\set{v} \notin q$, there exists $u' \in S$ such that $v$ and $u'$ are in the same block of $q$ and hence in the same connected component of $H_y(X, X_0)$, which shows that $(X, X_0)$ satisfies condition (vi) in the definition of $\mathcal{E}_x(p, S, j_1, j_2,  (r_i)_{i = 1} ^t)$ as well. We can easily verify that the rest of the conditions are satisfied too. Hence $(X, X_0) \in  \mathcal{E}_x(p, S, j_1, j_2,  (r_i)_{i = 1} ^t)$, which implies that $p \in A_x(S, j_1, j_2, (r_i)_{i = 1} ^t)$. 

We have thus shown that \fvsalgo\ computes $A_x(S, j_1, j_2, (r_i)_{i = 1} ^t)$ correctly when $x$ is a forget node. 


\textbf{Case 2. } Suppose that $x$ is an introduce node with child $y$. Let $x$ introduce the vertex $v$. First, if $v = v_0$ and $v \notin S$, then condition (v) in the definition of $\mathcal{E}_x(p, S, j_1, j_2,  (r_i)_{i = 1} ^t)$ cannot be satisfied, and hence $A_x(S, j_1, j_2, (r_i)_{i = 1} ^t) = \emptyset$. And in this case our algorithm correctly assigns the value $\emptyset$ to $A_x(S, j_1, j_2, (r_i)_{i = 1} ^t)$. So assume that either $v \neq v_0$ or $v \in S$ (or both). 
We now split the proof into two cases depending on $v \in S$ and $v \notin S$. 

\textbf{Case 2.1. } Assume first that $v \in S$. Let $p \in A_x(S, j_1, j_2, (r_i)_{i = 1} ^t)$. Then there exists $(X,X_0) \in \mathcal{E}_x(p, S, j_1, j_2,  (r_i)_{i = 1} ^t)$. Note that in the graph $G'_x$ and hence in the graph $H_x(X,X_0)$, the vertex $v$ is an isolated vertex. Thus $\set{v} \in p$, and thus $(X \setminus \set{v}, X_0) \in \mathcal{E}_y(p - v, S \setminus \set{v}, j_1 - 1, j_2, (r'_i)_{i = 1} ^t)$, where for each $i \in [t]$, let $r'_i = r_i - 1$ if $i \in c'(v)$ and $r'_i = r_i$ if $i \notin c'(v)$.  Therefore, $p - v \in A_y(S \setminus \set{v}, j_1 - 1, j_2, (r'_i)_{i = 1} ^t)$, and hence $p \in \ins(v, A_y(S \setminus \set{v}, j_1 - 1, j_2, (r'_i)_{i = 1} ^t))$. Conversely, let $q \in \ins(v, A_y(S \setminus \set{v}, j_1 - 1, j_2, (r'_i)_{i = 1} ^t))$. Then $q - v \in A_y(S \setminus \set{v}, j_1 - 1, j_2, (r'_i)_{i = 1} ^t)$, and so there exists $(X', X'_0) \in \mathcal{E}_y(q - v, S \setminus \set{v}, j_1 - 1, j_2, (r'_i)_{i = 1} ^t)$.  But then observe that $(X' \cup \set{v}, X'_0) \in \mathcal{E}_x(q, S, j_1, j_2, (r_i)_{i = 1} ^t)$. Thus, $q \in A_x(S, j_1, j_2, (r_i)_{i = 1} ^t)$.

\textbf{Case 2.2. } Assume now that $v \notin S$. Then $v \neq v_0$ and $S \subseteq B_y$. Again, as $v$ is an isolated vertex in the graph $G'_x$, notice that $(X, X_0) \in \mathcal{E}_x(p, S, j_1, j_2,  (r_i)_{i = 1} ^t)$ if and only if $(X, X_0) \in \mathcal{E}_y(p, S, j_1, j_2,  (r_i)_{i = 1} ^t)$. This implies that $p \in A_x(S, j_1, j_2, (r_i)_{i = 1} ^t)$ if and only if $p \in A_y(S, j_1, j_2, (r_i)_{i = 1} ^t)$. 

We have thus shown that \fvsalgo\ computes $A_x(S, j_1, j_2, (r_i)_{i = 1} ^t)$ correctly when $x$ is an introduce node. 

\textbf{Case 3. } Suppose that $x$ is an introduce edge node with child $y$. Let $x$ introduce the edge $e = uv$. 

\textbf{Case 3.1. } Let us first consider the case when $u = v_0$ and $v \in S$. 
Let $p \in A_x(S, j_1, j_2, (r_i)_{i = 1} ^t)$. 
Then there exists $(X, X_0) \in \mathcal{E}_x(p, S, j_1, j_2, (r_i)_{i = 1} ^t)$. If $e = uv = v_0v \notin X_0$, then $(X, X_0) \in \mathcal{E}_y(p, S, j_1, j_2, (r_i)_{i = 1} ^t)$, and hence $p \in A_y(S, j_1, j_2, (r_i)_{i = 1} ^t)$. 
Suppose $e = uv = v_0v \in X_0$. Then, $p \notin A_y(S, j_1, j_2, (r_i)_{i = 1} ^t)$. If $p \in A_y(S, j_1, j_2 - 1, (r_i)_{i = 1} ^t)$ then we are done. Otherwise, $\mathcal{E}_y(p, S, j_1, j_2 - 1, (r_i)_{i = 1} ^t) = \emptyset$. Let us consider the graph $H_y(X, X_0 \setminus \set{v_0 v})$. In this graph, the components which do contain $v_0$ or $v$ are the same as in $H_x(X, X_0)$, whereas, $v_0$ and $v$ are not in the same component, as that will imply, $(X, X_0 \setminus \set{v_0 v}) \in \mathcal{E}_y(p, S, j_1, j_2 - 1, (r_i)_{i = 1} ^t)$, which is a contradiction. Let, $C_v$ be the component containing $v$ and $C_{v_0}$ be the component containing. We notice that, $(C_v \cap S) \cup (C_{v_0} \cap S) =$ the block of $p$ containing $v_0$ and $v$. As the rest of the components when intersected with $S$ give rise to the blocks of $p$, so we have, $p \in \glue(v_0v, A_y(S, j_1, j_2 - 1, (r_i)_{i = 1} ^t))$. Conversely, let $q \in A_y(S, j_1, j_2, (r_i)_{i = 1} ^t) \cup \glue(v_0v, A_y(S, j_1, j_2 - 1, (r_i)_{i = 1} ^t))$. If $q \in A_y(S, j_1, j_2, (r_i)_{i = 1} ^t)$, then, $q \in A_x(S, j_1, j_2, (r_i)_{i = 1} ^t)$. Otherwise, $q \in \glue(v_0v, A_y(S, j_1, j_2 - 1, (r_i)_{i = 1} ^t))$. If, $q \in A_y(S, j_1, j_2 - 1, (r_i)_{i = 1} ^t)$, then there exists $(X',X'_0) \in \mathcal{E}_y(q, S, j_1, j_2 - 1, (r_i)_{i = 1} ^t)$,
and hence then $(X',X'_0 \cup \set{v_0v}) \in \mathcal{E}_x(q, S, j_1, j_2, (r_i)_{i = 1} ^t)$. Suppose, $q \notin A_y(S, j_1, j_2 - 1, (r_i)_{i = 1} ^t)$, then, $\mathcal{E}_y(q, S, j_1, j_2 - 1, (r_i)_{i = 1} ^t) = \emptyset$. This also implies that there exists, $q' \in A_y(S, j_1, j_2 - 1, (r_i)_{i = 1} ^t)$, such that, $\set{q} = \glue(v_0v,\set{q'})$. Let $(X'', X''_0) \in \mathcal{E}_y(q', S, j_1, j_2 - 1, (r_i)_{i = 1} ^t)$. In $H_y(X'',X''_0)$, $v_0$ and $v$ must be in separate components, as they are in separate blocks. Notice that $(X'',X''_0 \cup \set{v_0v}) \in \mathcal{E}_x(q, S, j_1, j_2, (r_i)_{i = 1} ^t)$ and thus, $q \in A_x(S, j_1, j_2, (r_i)_{i = 1} ^t)$.

The argument for the case when $v = v_0$ and $u \in S$ is symmetric.

\textbf{Case 3.2. } Let us now consider the case when $u, v \in S$. 
Let $p \in A_x(S, j_1, j_2, (r_i)_{i = 1} ^t)$. 
Then there exists $(X, X_0) \in \mathcal{E}_x(p, S, j_1, j_2, (r_i)_{i = 1} ^t)$. If $p \in A_y(S, j_1, j_2 - 1, (r_i)_{i = 1} ^t)$, then we are done. 
Otherwise, $\mathcal{E}_y(p, S, j_1, j_2 - 1, (r_i)_{i = 1} ^t) = \emptyset$. Let us consider the graph $H_y(X, X_0 \setminus \set{uv}$.
Let $p^{\star}$ be the partition of $S = X \cap B_x = X \cap B_y$ such that for $w_1, w_2 \in S$, $w_1$ and $w_2$ are in the same block of $p^{\star}$ if and only if they are in the same connected component of $H_y(X, X_0 \setminus \set{uv})$. Observe that the components of $H_y(X, X_0 \setminus \set{uv})$ which do not contain $u$ and $v$ are the same as in $H_x(X, X_0)$, whereas, $u$ and $v$ are not in the same component of the former as that would contradict that, $\mathcal{E}_y(p, S, j_1, j_2 - 1, (r_i)_{i = 1} ^t) = \emptyset$.
Then $p = \glue(uv, p^{\star})$, and thus $p \in \glue(uv, A_y(S, j_1, j_2 - 1, (r_i)_{i = 1} ^t))$. 
Conversely, let $q \in \glue(uv, A_y(S, j_1, j_2 - 1, (r_i)_{i = 1}^t))$. If $q \in A_y(S, j_1, j_2 - 1, (r_i)_{i = 1} ^t)$ then we are done. 
Otherwise, there is a finer partition, say $q^{\star}$, such that $\glue(uv, \{q^{\star}\}) = \{q\}$ and $q^{\star} \in A_y(S, j_1, j_2 - 1, (r_i)_{i = 1} ^t)$. Let $(Y, Y_0) \in \mathcal{E}_y(q^{\star}, S, j_1, j_2 - 1, (r_i)_{i = 1} ^t)$. In $H_y(Y, Y_0)$ the vertices $u$ and $v$ do not belong to same component as they are in different blocks of $q^{\star}$.  
However, $(Y, Y_0 \cup \set{uv}) \in \mathcal{E}_x(q, S, j_1, j_2, (r_i)_{i = 1} ^t)$.
Thus $q \in A_x(S, j_1, j_2, (r_i)_{i = 1} ^t)$.

\textbf{Case 3.3. } If none of the above cases hold, then either $u \notin S$ or $v \notin S$. In this case, for any $(X, X_0)$, we have $(X, X_0) \in \mathcal{E}_x(p, S, j_1, j_2, (r_i)_{i = 1} ^t)$ if and only if $uv \notin X_0$ and $(X, X_0) \in \mathcal{E}_y(p, S, j_1, j_2, (r_i)_{i = 1} ^t)$. Hence $p \in A_x(S, j_1, j_2, (r_i)_{i = 1} ^t$ if and only if $p \in A_y(S, j_1, j_2, (r_i)_{i = 1} ^t$. 

\textbf{Case 4. } Suppose that $x$ is a join node with children $y$ and $z$. 

Let $p \in A_x(S, j_1, j_2, (r_i)_{i = 1} ^t)$. 
Then there exists $(X, X_0) \in \mathcal{E}_x(p, S, j_1, j_2, (r_i)_{i = 1} ^t)$. We define $Y = X \cap V_y$, $Y_0 = X_0 \cap E_y$ and $Z = X \cap V_z$, $Z_0 = X_0 \cap E_z$. Let $j_1 ^y = \lvert Y \rvert$, $j_1 ^z = \lvert Z \rvert$, and, for each $i \in [t]$, $r_i ^y = c_i(Y)$ and $r_i ^z = c_i(Z)$. As $Y \cap B_y = Z \cap B_z = X \cap B_x = S$, we have, $j_1 ^y + j_1 ^z = j_1 + \card{S}$. And for each $i \in [t]$, we have, $r_i ^y + r_i ^z = r_i + c_i(S)$. Also we define, $j_2 ^y = \lvert Y_0 \rvert$ and $j_2 ^z = \lvert Z_0 \rvert$. Then, we have, $j_2 = j_2 ^y + j_2 ^z$. Let $p_y$ be the partiton of $S$ such that $(Y,Y_0) \in \mathcal{E}_y(p_y, S , j_1 ^y, j_2 ^y, (r_i ^y)_{i = 1} ^t)$. Similarly, let $p_z$ be the partition of $S$ such that $(Z, Z_0) \in \mathcal{E}_z(p_z, S, j_1 ^z, j_2 ^z, (r_i ^z)_{i = 1} ^t)$. Let $A \in p$. 
If $A \in p_y$ or $A \in p_z$, then we are done. If not, then we observe that the elements of $A$ are in the same component in $H_x(X,X_0)$. When we consider the graph $H_y(Y,Y_0)$, the elements of $A$ are in different components, since $A \notin p_y$. Similarly, for the graph $H_y(Z,Z_0)$. Thus, both $p_y$ and $p_z$ are finer than $p$. Thus, $A \in p_y \sqcup p_z$. Now, let $\hat{A} \in p_y \sqcup p_z$. If $\hat{A} \in p_y$ or $\hat{A} \in p_z$, then we are done. Otherwise, let, let $\{\hat{A_1} ^{y}, \hat{A_2} ^{y}, \ldots , \hat{A_{l_y}} ^{y} \}$ be the blocks of $p_y$ such that their union is $\hat{A}$ and let $\{\hat{A_1} ^{z}, \hat{A_2} ^{z}, \ldots , \hat{A_{l_z}} ^{z} \}$ be the blocks of $p_z$ such that their union is $\hat{A}$. Let, $w_1, w_2 \in A$. If they both belong to same block in $p_y$ or $p_z$, then are in the same component in $H_y(Y,Y_0)$ or $H_y(Y,Y_0)$ respectively and hence in the same component in $H_x(X,X_0)$. Otherwise, let's say $w_1 \in \hat{A_1} ^{y}$ and $w_2 \in \hat{A_2} ^{y}$ and we mark the elements in both $\hat{A_1} ^{y}$ and $\hat{A_2} ^{y}$. We, then collect all those blocks from $p_z$, which intersect $\hat{A_1} ^{y}$ and mark them. If they intersect $\hat{A_2} ^{y}$, then we are done, else we collect those blocks from $p_y$ which are intersected and mark the ones which were not already marked. Then we move to $p_z$ to check which are intersected and do the same. We keep doing this until we intersected $\hat{A_2} ^{y}$ . Notice that all the marked vertices are in the same component of $H_x(X,X_0)$. Thus, $\hat{A} \in p$.


Conversely, let $q \in \join(A_y(S, j^y_1, j^y_2, (r^y_i)_{i = 1} ^t), \\ A_z(S, j^z_1, j^z_2, (r^z_i)_{i = 1} ^t))$ for some $j^y_1, j^y_2, j^z_1, j^z_2, (r^y_i, r^z_i)_{i=1}^{t}$ with $j^y_1 + j^z_1 = j_1 + \card{S}, j^y_2 + j^z_2 = j_2, r^y_i + r^z_i = r_i + c_i(S)$. Then there exist $q_y \in A_y(S, j^y_1, j^y_2, (r^y_i)_{i = 1} ^t$ and $q_z \in A_z(S, j^z_1, j^z_2, (r^z_i)_{i = 1} ^t)$. Thus, there exists, $(Y, Y_0) \in \mathcal{E}_y(q_y, S, j_1 ^y, j_2 ^y, (r_i ^y)_{i = 1} ^t)$ and $(Z, Z_0) \in \mathcal{E}_z(q_z, S, j_1 ^z, j_2 ^z, (r_i ^z)_{i = 1} ^t)$. Let us define, $X = Y \cup Z$ and $X_0 = Y_0 \cup Z_0$. Thus we have, $(X, X_0) \in \mathcal{E}_x(p_x, S, j_1, j_2, (r_i)_{i = 1} ^t)$, and so, $q \in A_x(S, j_1, j_2, (r_i)_{i = 1} ^t)$. 

This completes the proof. 

\end{proof}
\end{sloppypar}

Finally, we analyse the runtime of \fvsalgo. 
\begin{lemma}
\label{lem:fvs_runtime}
\fvsalgo\ runs in time $n^{\cO(1)} \\ 2^{\cO({\twidth})}$. 
\end{lemma}

\begin{proof}
Given an instance $(G, c, (k_i)_{i = 1}^t)$ of \ffvs\ with a nice tree decomposition $(\mathcal{T}, (B_x)_{x \in V(T)})$ of $G$, we can construct $G', c'$ and a nice tree decomposition $(\mathcal{T}', (B_x)_{x \in V(\mathcal{T}')})$ in polynomial time. 
For each node $x \in \mathcal{T}'$, there are $2^{\card{B_x}} \leq 2^{\twidth + 1}$ many choices for $S \subseteq B_x$ and $n^{\cO(1)}$ choices for $ \\ (j_1, j_2, (r_i)_{i = 1}^{t})$. So the number of entries in the DP table is bounded by $\fvsruntime$. 
The time-consuming steps are the ones involving the computation of the entries $ \\ A_x(S, j_1, j_2, (r_i)_{i = 1}^{t})$. 
We have following cases:

\textbf{\underline{1. $x$ is a leaf node}}

As there is only one partition, and only $O(1)$ amount of entries in each cell, thus the total time spend in a leaf node is $n^{\cO(1)}$.

\textbf{\underline{2. $x$ is a forget node}}

Let $x$ forget the vertex $v$ and $y$ be the child of $x$.

	Let's say we have already applied Proposition~\ref{prop:reduce} to all the descendants of $x$. Then, $\card{A_y(.,.,.,.)} \leq 2^{\twidth}$. By, Proposition~\ref{prop:operations}, the time taken to evaluate $A_x(S, j_1, j_2, (r_i)_{i = 1}^{t})$ is $2^{\cO(\twidth)} \times {\twidth}^{\cO(1)}$. If $\card{A_x(S, j_1, j_2, (r_i)_{i = 1}^{t})} \leq 2^{\twidth}$, then we are done, otherwise we use Proposition~\ref{prop:reduce} on $ \\ A_x(S, j_1, j_2, (r_i)_{i = 1}^{t})$. For this, we need an additional $2^{\cO(\twidth)}$ time for this.

\textbf{\underline{2. $x$ is an introduce vertex node}}

Let $x$ introduce the vertex $v$ and $y$ be the child of $x$.

Let's say we have already applied Proposition~\ref{prop:reduce} to all the descendants of $x$. Then, $\card{A_y(.,.,.,.)} \leq 2^{\twidth}$. By, Proposition~\ref{prop:operations}, the time taken to evaluate $A_x(S, j_1, j_2, (r_i)_{i = 1}^{t})$ is $2^{\cO(\twidth)} \times {\twidth}^{\cO(1)}$. If $\card{A_x(S, j_1, j_2, (r_i)_{i = 1}^{t})} \leq 2^{\twidth}$, then we are done, otherwise we use Proposition~\ref{prop:reduce} on $ \\ A_x(S, j_1, j_2, (r_i)_{i = 1}^{t})$. For this, we need an additional $2^{\cO({\twidth})}$ time for this.

\textbf{\underline{4. $x$ is an introduce edge node}}

Same as above

\textbf{\underline{s. $x$ is a join node}}

Let $x$ be join node with $y$ and $z$ be the two children.

Assuming we have used, Proposition~\ref{prop:reduce} on all descendants of $x$, we will both $\card{A_y(.,.,.,.)}$ and $\card{A_z(.,.,.,.)}$ to be at most $2^{\twidth}$. The time taken for join operation is $2^{\cO({\twidth})} {\twidth}^{\cO(1)}$. If $\card{A_x(S, j_1, j_2, (r_i)_{i = 1}^{t})} \leq 2^{\twidth}$, then we are done, otherwise we use Proposition~\ref{prop:reduce} on $A_x(S, j_1, j_2, (r_i)_{i = 1}^{t})$. For this, we need an additional $2^{\cO({\twidth})}$ time for this.

Thus, the total time taken in each entry of the DP table is bounded by $2^{\cO({\twidth})} {\twidth}^{\cO(1)}$. Thus, the total time taken is $n^{\cO(1)} 2^{\cO({\twidth})}$.

\end{proof}

Lemmas~\ref{lem:fvs_correctness}-\ref{lem:fvs_runtime} together prove Theorem~\ref{thm:fvs}. 


\subsubsection{\ffvsfull\ Parameterized by Total Color Budget}
\label{sec:fvs-fpt}
In this section, we design an algorithm for \ffvs, parameterized by the total colour budget. Specifically, we prove the following theorem. 

\begin{theorem}
\label{thm:fvs-fpt}
\ffvsfull\ admits an algorithm that runs in time $2^{\cO(k)} n^{\cO(1)}$, where $n$ is the number of vertices in the input graph, and $k = \sum_{i = 1}^{t} k_i$ is the total colour budget. 
\end{theorem}

To prove Theorem~\ref{thm:fvs-fpt}, we use the following approximation algorithm for \FVS\ due to Bafna et al.~\cite{DBLP:journals/siamdm/BafnaBF99}.

\begin{proposition}
\cite{DBLP:journals/siamdm/BafnaBF99}
\label{prop:fvs-approx}
There is an algorithm that, given a graph $G$ as input, runs in polynomial time, and returns a set $F_{apx} \subseteq V(G)$ such that $F_{apx}$ is a feedback vertex set of $G$ and $\card{F_{apx}} \leq 2 \cdot \mathsf{optfvs}(G)$, where $\mathsf{optfvs}(G)$ is the minimum size of a feedback vertex set of $G$. 
\end{proposition} 

We now design our FPT algorithm, which we call \fvsalgofpt. Given a $t$-coloured graph $(G, c)$ and a $t$-tuple $\mathbb{T} = (k_1, k_2,\ldots, k_t)$, \fvsalgofpt works as follows. 

\begin{description}
\item[Step 1.] We first invoke the algorithm in Proposition~\ref{prop:fvs-approx} on the input $G$; and let $F_{apx} \subseteq V(G)$ be the fvs returned by algorithm in Proposition~\ref{prop:fvs-approx}. If $\card{F_{apx}} > 2 \sum_{i = 1}^t k_i$, then we return that $(G, c)$ has no $\mathbb{T}$-fair fvs. 

\item[Step 2.] We construct the (acyclic) graph $G - F_{apx}$, and construct a tree decomposition of $G - F_{apx}$ of width $1$. Then, to every bag in this tree decomposition, we add the vertices in $F_{apx}$ to obtain a tree decomposition
 of $G$. Finally, we make the tree decomposition nice. Let $(\mathcal{T}, (B_x)_{x \in V(\mathcal{T})})$ be resulting nice tree decomposition of $G$. (Notice that the tree decomposition $ \\ (\mathcal{T}, (B_x)_{x \in V(\mathcal{T})})$ has width $\twidth = 1 + \card{F_{apx}} \leq 2 \sum_{i = 1}^t k_i$.)
 
 \item[Step 3.] We invoke the algorithm in Theorem~\ref{thm:fvs} on input $ \\ (G, c, \mathbb{T})$ along with the nice tree decomposition $ \\ (\mathcal{T}, (B_x)_{x \in V(\mathcal{T})})$. If the algorithm in Theorem~\ref{thm:fvs} returns yes, then we return that $(G, c)$ has a $\mathbb{T}$-fair fvs, and otherwise we we return that $(G, c)$ has no $\mathbb{T}$-fair fvs. 
\end{description}
We now argue that \fvsalgofpt\ is correct and analyse its runtime. 
\begin{lemma}
\label{lem:correctness-fvsfpt}
\fvsalgofpt\ is correct. 
\end{lemma}

\begin{proof}
To see the correctness of Step 1, observe that if $ \\ \card{F_{apx}} > 2 \sum_{i = 1}^t k_i$, then by Proposition~\ref{prop:fvs-approx}, $ \\ \mathsf{optfvs}(G) > \sum_{i = 1}^t k_i$. That is, every minimum-sized fvs of $G$ has size at least $1 + \sum_{i = 1}^t k_i$, which implies that $(G, c)$ has no $\mathbb{T}$-fair fvs. And the correctness of Step 3 follows from Theorem~\ref{thm:fvs}. 
\end{proof}

\begin{lemma}
\label{lem:runtime-fvsfpt}
\fvsalgofpt\ runs in time $ \\ 2^{\cO(\sum_{i = 1}^t k_i)} n^{\cO(1)}$. 
\end{lemma}

\begin{proof}
Notice that Steps 1 and 2 of \fvsalgofpt\ take only polynomial time. The time taken for Step 1 follows from Proposition~\ref{prop:fvs-approx}. As for Step 2, notice that the graph $G - F_{apx}$ is acyclic, and hence, in polynomial time we can construct a tree decomposition of $G - F_{apx}$ of width $1$. (For example, for each connected component $C$ (which is a tree) of $G - F_{apx}$, we root $C$ at an arbitrary vertex $x_C \in V(C)$, and construct a tree decomposition $\mathcal{T}^C, (B_x)_{x \in V(\mathcal{T}^C)}$ as follows: We take $\mathcal{T}^C = C$, $B_{x_C} = \set{x_C}$ and $B_x = \set{x, y}$ for $x \in V(C)$ with parent $y$. We can then connect the trees $\mathcal{T}^C$ for different components $C$ by introducing a new common root and making it adjacent to $x_C$ for every component $C$. This can be done in polynomial time.) The remaining part of Step 2 only involves adding $F_{apx}$ to every bag in the tree decomposition of $G - F_{apx}$, and making the decomposition nice. This can also be done in polynomial time. By Theorem~\ref{thm:fvs}, Step 3 takes time $2^{\cO({\twidth})} n^{\cO(1)} = 2^{\cO(\sum_{i = 1}^t k_i)} n^{\cO(1)}$, as $\twidth = 1 + \card{F_{apx}} \leq 2 \sum_{i = 1}^t k_i$. Thus the total time taken by the algorithm is bounded by $2^{\cO(\sum_{i = 1}^t k_i)} n^{\cO(1)}$. 
\end{proof}

Lemmas~\ref{lem:correctness-fvsfpt} and ~\ref{lem:runtime-fvsfpt} together prove Theorem~\ref{thm:fvs-fpt}. 

\subsection{Algorithms for $(\alpha, \beta) \mbox{-} \mathbb{T}$-{\sc Fair $\Pi$}}\label{sec:relaxation}

As discussed earlier, $\mathbb{T}$-{\sc Fair $\Pi$} imposes a rather strict condition on the solution---the solution must contain \emph{exactly} $k_i$ vertices of colour $i$ for each $i$. Such solutions, however, may not exist. So we must relax this strict condition. 
Let $k = \sum_{i=1}^t k_i$.
In particular,  we would want to answer questions such as these: For fixed constants $\alpha, \beta$, does $G$ have a vertex cover $S \subseteq V(G)$ of size at most $k$ that satisfies the following 
for every colour $i \in [t]$?
\[
\alpha \frac{c_i(V(G))}{\card{V(G)}} \leq \frac{c_i(S)}{\card{S}} \leq \beta \frac{c_i(V(G))}{\card{V(G)}}
\]
Recall that for $i \in [t]$ and $X \subseteq V(G)$, $c_i(X)$ denotes the number of vertices of colour $i$ in $X$. 
In this section, we show how we can use the algorithms for \fairvc\ and \ffvs\ to answer such questions. 
To illustrate the ideas, we confine our discussion to {\sc Vertex Cover}. But analogous results hold for {\sc $(\alpha, \beta) \mbox{-} \mathbb T$-Fair Feedback Vertex Set} as well as {\sc $(\alpha, \beta) \mbox{-} \mathbb T$-Fair $\Pi$} for any classic optimization problem $\Pi$.

Consider a $t$-coloured graph $G$ and a tuple $\mathbb{T} = (k_1, k_2,\ldots, k_t)$. Let $k = \sum_{i=1}^t k_i$. For $\alpha, \beta$, where $0 \leq \alpha \leq 1 \leq \beta$, we say that $S \subseteq V(G)$ is $(\alpha, \beta) \mbox{-} \mathbb{T}$-fair if $|S| \leq k$ and $\alpha k_i \leq c_i(S) \leq \beta k_i$ for every $i \in [t]$.  That is, the number of vertices of colour $i$ in $S$ is between $\alpha k_i$ and $\beta k_i$. 

Observe that for $\mathbb{T} = (k_1, k_2,\ldots, k_t)$, a set of vertices $S \subseteq V(G)$ is $(\alpha, \beta) \mbox{-} \mathbb{T}$ fair if and only if $S$ is $\mathbb{T}'$-fair for some $\mathbb{T}' = (k'_1, k'_2,\ldots, k'_t)$, where $\alpha k_i \leq k'_i \leq \beta k_i$ for every $i \in [t]$ and $\sum_{i=1}^t k'_i \leq k$. And for each $i$, there are at most $\beta k_i - \alpha k_i + 1$ possible choices for $k'_i$, which implies that the number of possible choices for $\mathbb{T}'$ is at most $\Pi_{i \in [t]} ((\beta - \alpha) k_i + 1))$. This observation immediately leads to the following lemma.

\vcAlphaBetaThm


 \begin{proof}
	 Given an instance $(G, c, \mathbb{T} = (k_1, k_2,\ldots, k_t))$ of $ \\ (\alpha, \beta) \mbox{-} \mathbb{T}$-{\sc Fair VC}, we proceed as follows. For each $\mathbb{T}' = (k'_1, k'_2,\ldots, k'_t)$ such that $\alpha k_i \leq k'_i \leq \beta k_i$ for every $i \in [t]$, we check if $(G, c)$ has a $\mathbb{T}'$-fair vertex cover; if $(G, c)$ has a $\mathbb{T}'$-fair vertex cover for some $\mathbb{T}'$, we return that $(G, c)$ has an $ (\alpha, \beta) \mbox{-} \mathbb{T}$-fair vertex cover; otherwise, we return no. As $\mathbb{T}'$ has at most $\Pi_{i \in [t]} ((\beta - \alpha) k_i + 1))$ choices, the lemma follows. 
 \end{proof}

A similar result can be stated for {\sc $(\alpha, \beta) \mbox{-} \mathbb T$-Fair $\Pi$} for any classic optimization problem $\Pi$, and in particular for $(\alpha, \beta) \mbox{-} \mathbb{T}$-{\sc Fair FVS}

\alphaBetaThm
\fvsAlphaBetaThm

\section{Conclusion}\label{sec:conclusion}

In this paper, we defined a notion of \emph{unbiased} solutions to combinatorial problem. We introduced a definition to formally derive the unbiased variant of a classic combinatorial problem. We then explored the variants of {\sc Vertex Cover} and {\sc Feedback Vertex Set} in the papradigm of Parameterized Complexity, and give efficient algorithms for them. The natural next step is to explore the parameterized complexity of the unbiased variant of other well-studied problems, such as {\sc Planar Dominating Set}, {\sc Odd Cycle Transversal}, {\sc Matching under Preferences} and many others. Another direction worth exploring is the notion of an unbiased approximation solution, which requires some further research.

\bibliography{example_paper, newbiblio}
\bibliographystyle{splncs04}

\end{document}